\documentclass{article}

\usepackage[top=1.5cm,bottom=1.5cm,left=1.5cm,right=1.5cm]{geometry}
\usepackage[margin=1.cm]{caption}

\usepackage[english]{babel}
\usepackage[pdftex]{graphicx}
\usepackage{authblk}
\usepackage{amsmath}
\usepackage{amsthm}
\usepackage{amssymb}
\usepackage{anyfontsize}
\usepackage{enumitem}
\usepackage{algorithmicx}
\usepackage{algorithm}
\usepackage{algpseudocode}
\usepackage{bbm}
\usepackage{commath}
\usepackage{centernot} 
\usepackage{mathrsfs}
\usepackage{bm}
\usepackage{hyperref}
\usepackage[numbers,sort]{natbib}
\usepackage{doi}

\DeclareFontFamily{OT1}{pzc}{}
\DeclareFontShape{OT1}{pzc}{m}{it}{<-> s * [1.10] pzcmi7t}{}
\DeclareMathAlphabet{\mathpzc}{OT1}{pzc}{m}{it}

\newtheorem{theorem}{Theorem}

\newtheorem{lemma}{Lemma}
\newtheorem{definition}{Definition}

\newtheorem{proposition}{Proposition}

\title{Fault-Tolerant Quantum Error Correction for non-Abelian Anyons}
\date{\today}
\author{Guillaume Dauphinais and David Poulin}
\affil{Institut quantique \& D\'epartement de physique, Universit\'e de Sherbrooke}

\begin{document}

\maketitle

%\begin{abstract}
\abstract{While topological quantum computation is intrinsically fault-tolerant at zero temperature, it loses its topological protection at any finite temperature. We present a scheme to protect the information stored in a system supporting non-cyclic anyons against thermal and measurement errors. The correction procedure builds on the work of G\'acs \cite{Gacs_86} and Harrington \cite{Harrington_04} and operates as a local cellular automaton. In contrast to previously studied schemes, our scheme is valid for both abelian and non-abelian anyons and accounts for measurement errors. We analytically prove the existence of a fault-tolerant threshold for a certain class of non-Abelian anyon models, and numerically simulate the procedure for the specific example of Ising anyons. The result of our simulations are consistent with a threshold between $10^{-4}$ and $10^{-3}$.}
%\end{abstract}
%

\section{Introduction}
\label{sec_intro}

Non-abelian anyons are hypothetical particles with very exotic properties that defy intuition but that are nonetheless permitted by known laws of physics.  These particles have drawn much interest due to their suspected existence in two-dimensional condensed matter systems and for their potential applications in quantum computation \cite{Das_Sarma_15, Freedman_06, Bonderson_10, Kitaev_03, Nayak_08, Trebst_09, Stern_10}. In particular, a quantum computation can in principle be realized by braiding and fusing certain non-abelian anyons~\cite{Freedman_02, Freedman_02_2}. These operations are  intrinsically robust due to their topological nature. Because systems supporting anyonic excitations have a spectral gap $\Delta$, interactions between particles are short-ranged, so the details of braiding operations should not matter as long as the anyons are kept at a sufficiently large distance $\sim 1/\Delta$  from each other. Moreover, the accuracy of braiding operations can be made arbitrarily good by increasing the distance between computational anyons.

The spectral gap also offers some protection against thermal excitations. Provided the system is kept at a temperature $T$ lower than the spectral gap, the density of thermal excitations is suppressed by an exponential Boltzmann factor $e^{-\Delta/T}$. In contrast to the topological protection, however, this thermal protection is not scalable: thermal excitations do appear at constant density for any non-zero temperatures and so their presence is unavoidable as the size of the computation increases. Thermally activated anyons can corrupt the encoded data by braiding or fusing with the computational anyons \cite{Budich_12, Schmidt_12, Pedrocchi_15, Pedrocchi_15_2, Goldstein_11, Konschelle_13, Rainis_12,Wootton_14}. Schemes have been proposed to physically boost the thermal protection \cite{Hamma_09, Chesi_10, Pedrocchi_13}, but none of them offers a scalable solution \cite{Landon_13, Landon_15}. It thus appears necessary to supplement topological quantum computations with some form of quantum error correction.

Error correction in abelian anyonic models is intrinsically linked to topological quantum error correction with the toric code and has thus been studied extensively \cite{Dennis_02, Duclos_10, Duclos_10_2, Wang_10, Wootton_12, Bravyi_Haah_13, Anwar_13, Hutter_14, Bravyi_14, Herold_14, Wootton_15_2, Andrist_15}. There, it is possible to model the different thermal processes phenomenologically using a particle creation rate and a diffusion rate. Error correction monitors the presence of these thermal excitations by periodically measuring the topological charge at every lattice site. A {\em decoding algorithm} is used to statistically infer the homology of each particle's world-line from these snapshots, thus enabling the recovery of the topological information. It is now well established that these systems possess a threshold:  below a critical ``temperature'', the logical failure rate can be suppressed to arbitrarily low values by increasing the system size. 

These error-correction studies assume that the topological charge measurements are perfect. In a realistic setting, a  measurement can report the wrong charge---e.g., report a charge when the site if empty or fails to report a charge---and, moreover, it can introduce additional errors. The ability to protect a topologically ordered system using such noisy charge measurements is intrinsically linked to fault-tolerant topological quantum error correction, where fault-tolerance refers to the ability to combat errors with noisy instruments. Again, for abelian anyons, this problem has been studied extensively \cite{Dennis_02, Wang_10, Fowler_09, Duclos_14, Watson_15, Fowler_15, Harrington_04, Herold_15} and is known to possess a fault-tolerance threshold. 

The theory of error correction for non-abelian anyons is in contrast far less developed. Specific examples of error correcting schemes for Ising anyons~\cite{Brell_14}, the $\Phi - \Lambda$ model~\cite{Wootton_14} and Fibonacci anyons~\cite{Burton_15} have been investigated numerically and found to display threshold behaviours. Additionally, greedy hard-decision renormalization group decoders can error-correct any systems giving rise to anyonic excitations~\cite{Wootton_15,Hutter_15_2}.  However, none of these studies have considered the case where the charge measurements are faulty, a serious complication for all the previous methods. 

In the present work, we generalize a fault-tolerant scheme introduced by Harrington~\cite{Harrington_04} for the toric-code to the setting of non-cyclic modular anyons.  Although non-cyclic anyons are non-abelian in general, they have the property of flowing towards an abelian model under fusion, in a sense that will be made precise below. Thus, our work establishes that fault-tolerant quantum computation can be realized with non-abelian anyons. Moreover, this decoding algorithm builds on the work of G\'acs \cite{Gacs_86,Gray_01} and combines ideas of cellular automatons and renormalization in two spatial dimensions. The charge measurement information is handled by local rules on a two-dimensional lattice and does not require global processing. 

The basic idea of all topological decoding algorithm is to pair up the thermal anyons, bring the anyons of a pair together and hope that they fuse to the vacuum. The details of how the pairs are chosen and how the anyons of a pair are brought together is specific to each decoding algorithm. Non-abelian anyons present novel obstacles to this general decoding strategy. First, the fusion process is intrinsically irreversible for non-abelian anyons. In particular, when two anyons of `opposite' topological charges $a$ and $\bar a$ are brought together, they may fuse to a non-trivial charge. Observing the outcome of this fusion is an irreversible process (collapse of the wave-function), so the error-correction process itself could introduce physically irreversible changes to the system. Note that such a non-trivial fusion also offers a possible advantage: it provides a clear indication that the two charges were not created from the same thermal process, and this extra piece of information could be used by a decoder. Second, there are many more ways in which small errors involving non-abelian anyons can build up to a larger error. This is illustrated on Fig.\ref{fig:BuildUp}. In contrast to abelian anyons, it is not necessary for particles to meet head-on to `sew' small errors into a larger one; a simple braid will suffice. Because there are many more ways to braid than to fuse, this extra entropic contribution could favour a disordered phase and prevent reliable topological information processing. As we will discuss in Sec. \ref*{sec:defs}, the non-abelian nature of the excitations also presents significant additional obstacles to the analysis of the error correction procedure itself. Despite these complications, we will show that the decoding algorithm possesses a threshold. 

\begin{figure}
	\begin{center}
	\includegraphics[scale=.5]{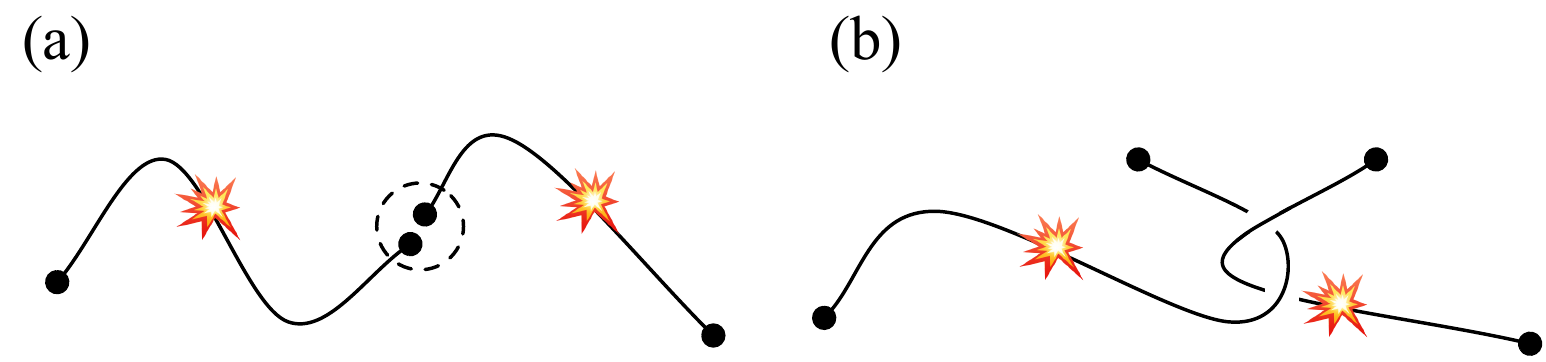}
	\end{center}
	\caption{Example of two small errors building up to a larger one. (a) A pair of anyons originating from distinct fission processes collide and fuse into the vacuum. (b)  A pair of non-abelian anyons originating from distinct fission processes braids. }
	\label{fig:BuildUp}
\end{figure}

In Sec. \ref{sec_anyon_theory}, we review basic concepts of the algebraic theory of anyons. The noise model we consider is presented in Sec. \ref{sec_noise}, together with a classification of the error events in terms of renormalization levels. Renormalized error rates are defined and shown to decrease doubly exponentially with the renormalization level. These sections essentially follow Harrington's work. The correction algorithm is presented in Sec. \ref{sec_corr_algo}. Various definitions and concepts are introduced in Sec. \ref{sec:defs} which are used to prove in Sec. \ref{sec_threshold} the existence of a threshold error rate below which the memory lifetime can be increased to arbitrarily long times by increasing the system size. Sec. \ref{sec_numerical} presents numerical simulations for a system of Ising anyons on a torus, which suggest a threshold between $10^{-3}$ - $10^{-4}$. Concluding remarks are made in Sec. \ref{sec_discussion}.

\section{Basic Data of the Algebraic Theory of Anyons}
\label{sec_anyon_theory}

In this section we present the key concepts of the algebraic theory of anyons, more details can be found in \cite{Kitaev_06} and \cite{Bonderson_thesis}. We assume that the system under study is defined on a two-dimensional surface, has short-range interaction and is gapped. Excitations are well localized and it is possible to change their position by applying a suitable local operator along an arbitrary path. 

Excitations are classified into superselection sectors. A sector consists of states (excitations) that can be transformed into each other by the application of local  operators --- operators acting on a homologically trivial region. We also assume that it is possible to measure the topological sector of the excitations (charge measurement).

Each excitation (which we also refer to as particle or anyon) is described by a label $l$ denoting its superselection sectors. The various possible labels are also referred to as topological charges.  The trivial charge (the vacuum) is denoted herein by $1$. The action of bringing two excitations together and to measure the resulting charge is called a fusion process, while a splitting process can be seen as its conjugate process where a single charge is `split' into two charges. With every elementary fusion event is associated the Hilbert space $V^{ab}_{c}$, the space of states of particles of charge $a$ and $b$ restricted to have a total charge $c$. Correspondingly, the Hilbert space $V^{c}_{ab}$ is the space of states of charge $c$ restricted to come from the fusion of two charges of type $a$ and $b$. Orthonormal basis vectors for the splitting and fusion spaces are labelled by $\vert a, b; c, \psi \rangle \in V^{ab}_{c}$ and $\langle a, b, ; c, \psi \vert \in V_{ab}^{c}$, where $\psi$ denotes the fusion channel. The splitting and fusion operations can be graphically represented by using trivalent vertices with the corresponding charge and fusion channel labels:
\begin{equation}
	\parbox{0.6 \textwidth}{
		\includegraphics[scale=.9]{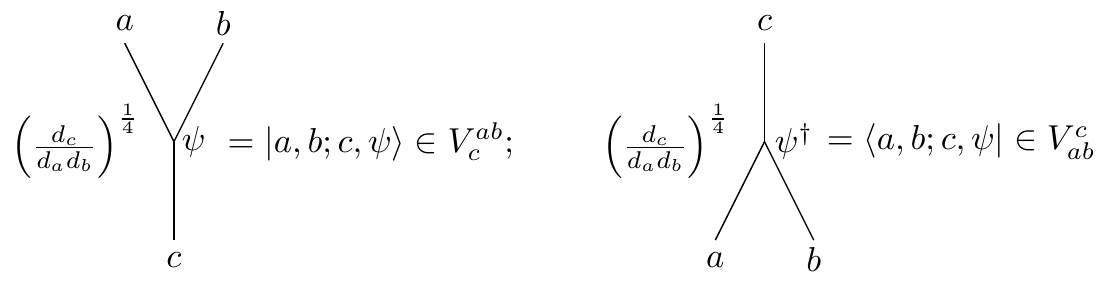}
		}
		\label{eqn_fusion_space}
\end{equation}
The normalization factors $d_{x}$'s are introduced so that the diagrams are in the isotopy-invariant convention --- bending lines and rotating part of the diagrams change the amplitudes only by unitary transformations.

Fusion multiplicities are defined as $N_{ab}^{c} = \text{dim } V^{c}_{ab} = \text{dim } V^{ab}_{c}$. Furthermore, we will assume that the duality axiom holds: for every charge $a$, there exist a unique charge $\bar{a}$ (also called $a$'s antiparticle) such that $N_{a \bar{a}}^{1} = 1$. It is possible that $\bar{a} = a$, such anyons are called self-dual.

Consider the case where charges $a, b$ and $c$ are restricted to have total charge $d$. There is freedom in terms of the order in which the particles are fused together, corresponding to different orthonormal bases in which a state is represented. Such a change of basis is described by an $F$-matrix defined as
\begin{equation}
	\parbox{0.5\textwidth}{
		\includegraphics[scale=.8]{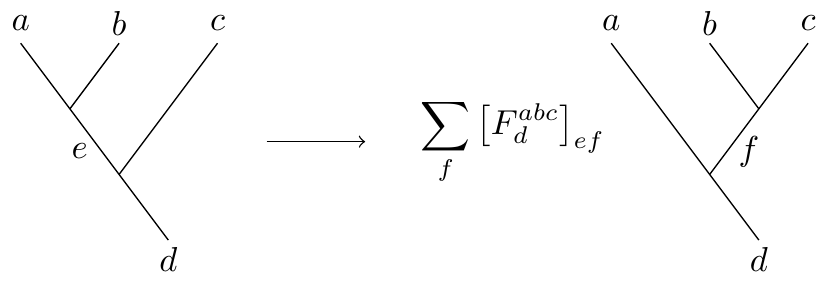}
	}
	\label{eqn_f_move}
\end{equation}

Another useful quantity is the quantum dimension of a particle, $d_{a}$, which is defined by $d_{a} = \vert \left[ F^{a \bar{a} a}_{a}\right]_{1 1} \vert^{-1}$. Physically, $d_{a}^{2}$ represents the inverse of the probability that two anyons $a$ and $\bar{a}$ created from different pairs from the vacuum fuse together to the vacuum.

The exchange of particles in the clockwise direction is defined in terms of braiding operations which can be diagrammatically represented by
\begin{equation}
	\parbox{0.4\textwidth}{
		\includegraphics[scale=.8]{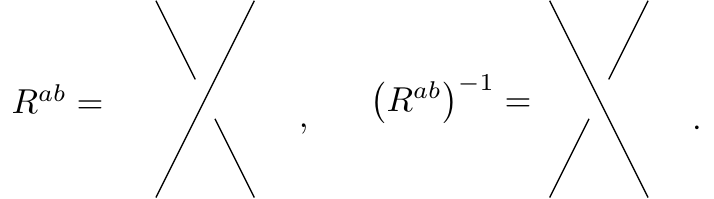}
	}
	\label{eqn_braiding}
\end{equation}
The action of the braiding operators on the fusion spaces can be completely determined in terms of the $R$-symbols representing the unitary operator whose action is to exchange two anyons (either in a clockwise or counterclockwise fashion) stemming from a specific fusion channel:
\begin{equation}
	\parbox{0.3\textwidth}{
		\includegraphics[scale=.8]{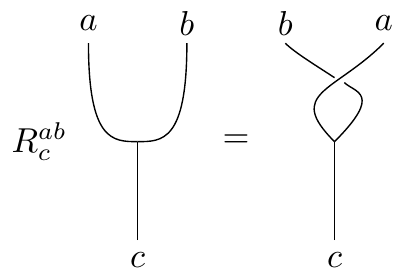}
	}
	\label{eqn_R_move}
\end{equation}

An arbitrary braiding between two charges that were not created together can be expressed in terms of $F$-moves combined with braiding operations $R_{c}^{ab}$ and $\left( R_{c}^{ab} \right)^{-1}$. In general, the braiding of two anyons applies not just a phase, but can change the fusion states of the anyons. If, for each charge $x \neq 1$, there is some label $a$ such that $R^{ax}R^{xa}$ is different than the identity operator, then braiding is said to be {\it non-degenerate}. Systems giving rise to anyonic excitations can be defined on closed surfaces of genus higher than $0$ if and only if its braiding is non-degenerate~\cite{Kitaev_06}. An anyonic theory which is non-degenerate is also said to be {\it modular}.

Systems of anyons obey the following property, which implies the Yang-Baxter relations \cite{MacLane} :
\begin{equation}
	\parbox{0.2\textwidth}{
		\includegraphics[scale=.3]{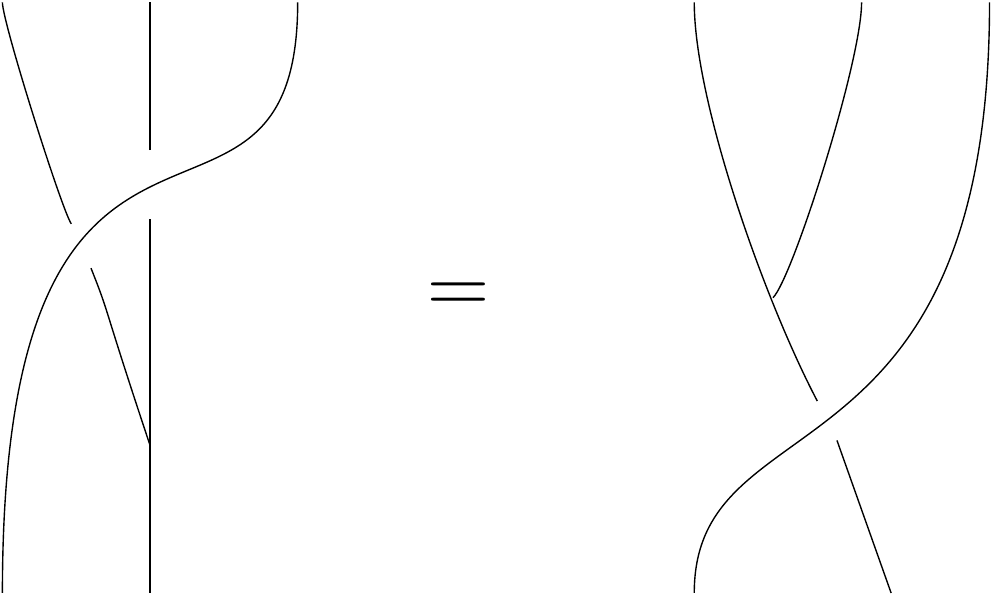}
	}
	\label{eqn_yang_baxter}
\end{equation}
That is, one can freely pass lines above or below vertices, without introducing any phase factor.

The inner product of two states, say $\langle \phi \vert \psi \rangle$ is diagrammatically represented by connecting the corresponding leaves of the fusion trees associated with the two states $\langle \phi \vert$ and $\vert \psi \rangle$, and is diagrammatically given by:
\begin{equation}
 	\parbox{0.3\textwidth}{
		\includegraphics[scale=.8]{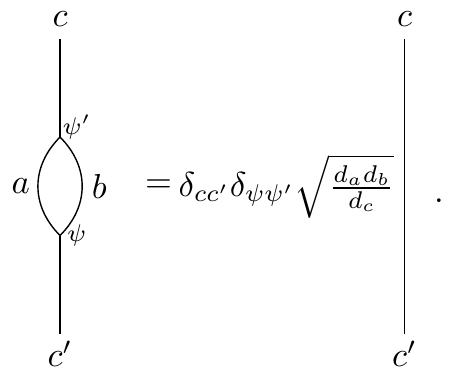}
  }
\label{eqn_loops}
\end{equation}

The topological S-matrix is defined by
\begin{equation}
	\parbox{0.3\textwidth}{
		\includegraphics[scale=1.0]{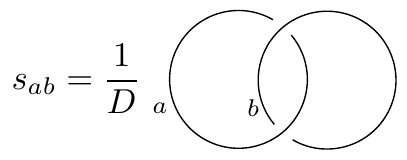}
	}
	\label{eqn_S_matrix}
\end{equation}
where $D = \sum_{a} \sqrt{d_{a}^{2}}$ is such that for modular theories the S-matrix is unitary.

For modular theories, a collective charge projector can be defined~\cite{Levaillant_15} by
\begin{equation}
	\parbox{0.3\textwidth}{
		\includegraphics[scale=1.0]{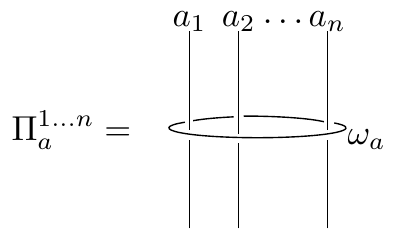}
	}
	\label{eqn_projector}
\end{equation}
where $\omega_{a}$ is given by
\begin{equation}
	\parbox{0.3\textwidth}{
		\includegraphics[scale=1.0]{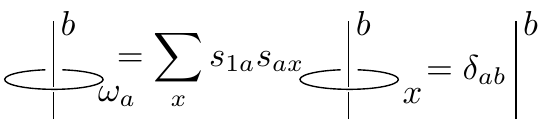}
	}
	\label{w_loop}
\end{equation}
The effect of $\prod_{a}^{1 \dots n}$ is to project the anyons it encompasses into the collective topological charge $a$, and the probability to measure charge $a$ is given by 
\begin{equation}
\langle \psi \vert \Pi_{a}^{1 \dots n} \vert \psi \rangle,
\label{eqn_meas_prob}
\end{equation}
 where $\vert \psi \rangle$ is the state of the system~\cite{Bonderson_08_a,Bonderson_08_b}.
%\subsection{Encoded information in the model}

The ground state degeneracy of a system giving rise to modular anyons living on an unpunctured 2-dimensional surface $\Sigma$ of genus $g$ is given by $ \sum_{a} s_{1a}^{2-2g}$ and consists of the vacuum \cite{Verlinde_88}. In the case of a torus, this last expression reduces to $\vert \mathcal{A} \vert$, the number of different superselection sectors in the anyonic model $\mathcal{A}$. The ground space can be seen as a code space \cite{Kitaev_03}, and logical operators consist of creating particle and antiparticle from the vacuum, performing an homologically non-trivial loop and fusing  the particle and antiparticle back to the vacuum. It is also in general possible to leave the ground space by doing such operations, although in this case the particle/antiparticle pairs fuse to give a non-trivial charge.

\subsection{Definition of Non-Cyclic Anyon Models}
\label{sec_cyclic_anyons}

In view of the algorithm presented below, it is relevant to define a family of anyons called {\it non-cyclic} anyonic models. An anyonic model is non-cyclic if and only if
\begin{equation}
\prod_{i=1}^{n} N_{x_{i} \bar{x}_{i}}^{x_{i+1}}= 0,
\label{eqn_condition1}
\end{equation}
for any value of $n$ and for any sequence $\{ x_{1}, x_{2}, \dots, x_{n}, x_{n+1} = x_{1} \}$ for which $x_{1} \neq 1$. An anyonic model which is not non-cyclic is called cyclic. From this definition it is clear that abelian models are all non-cyclic.

A perhaps more enlightening equivalent definition can be formulated using a graphical representation of fusion rules. For any anyonic model $\mathcal{A}$, construct the following directed graph $G_{\mathcal{A}} = (V,E)$.  Associate a vertex $v\in V$ with each couple of particle/anti-particle denoted by $v = \left( a, \bar{a} \right)$. For all charges $a$ different than the vacuum, if $N_{a \bar{a}}^{b} > 0$ then add a directed edge from node $\left( a, \bar{a} \right)$ to node $\left( b, \bar{b} \right)$. Then, $\mathcal{A}$ is said to be non-cyclic if $G_{\mathcal{A}}$ does not contain cycles. Otherwise, $\mathcal{A}$ is cyclic. If $\mathcal{A}$ is non-cyclic, denote by $diam \left( G_{\mathcal{A}} \right)$ the diameter of $G_{\mathcal{A}}$. Examples of such graphs are shown on Figure~\ref{fig_graph_anyon}.

\begin{figure}
	\begin{center}
	\includegraphics[scale=.5]{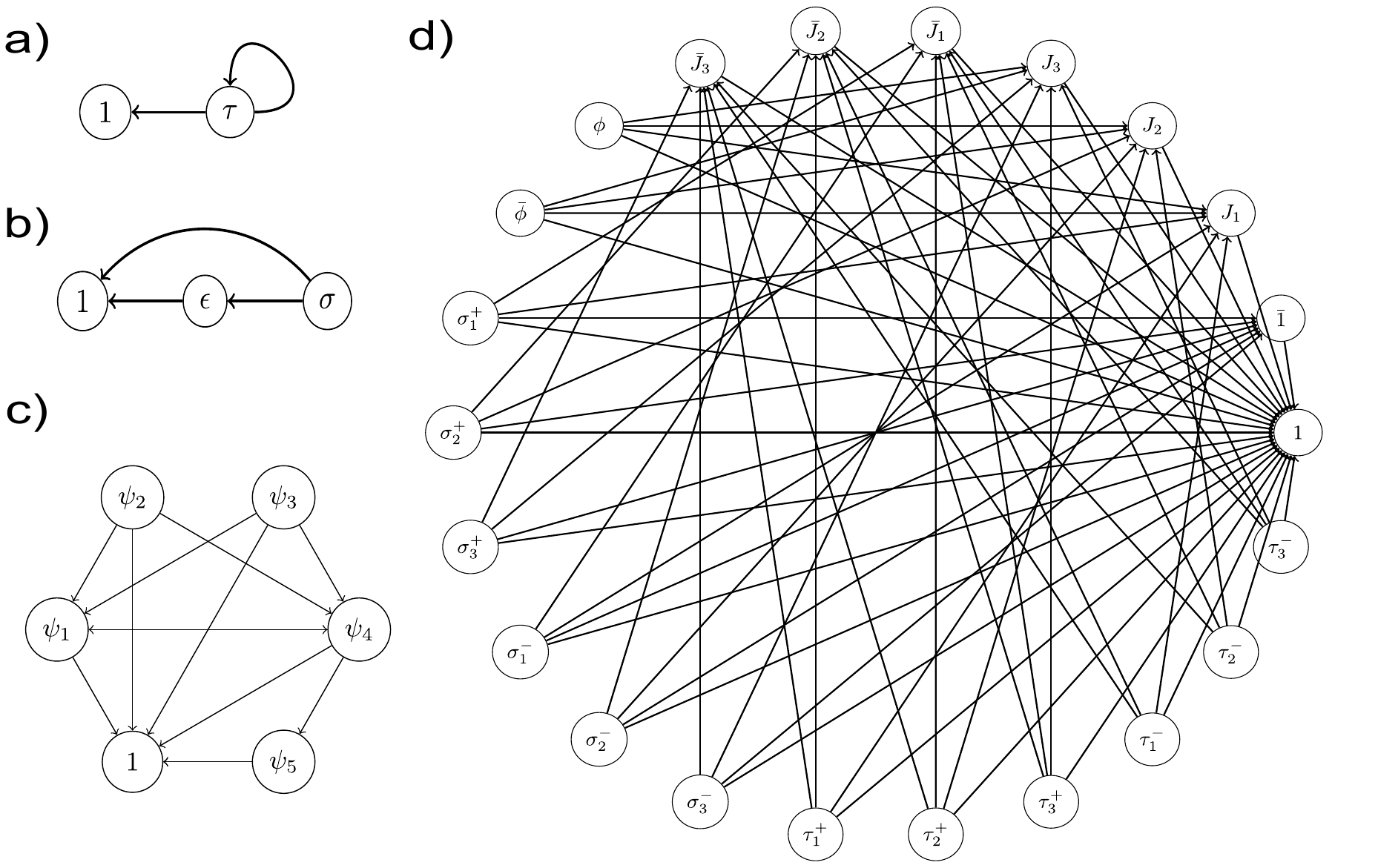}
	\end{center}
	\caption{Examples of the graphs associated with Fibonacci anyons (a), Ising anyons (b), the SO$\left( 5 \right)_{2}$ anyons (c) and anyons stemming from the quantum double of the quaternion group $D \left( \bar{H} \right)$ (d). Among these examples, both Ising and $D \left( \bar{H} \right) $ are non-cyclic.}
	\label{fig_graph_anyon}
\end{figure}

Although for anyonic models containing less than $5$ charges only abelian models and the Ising model (and closely related ones) are non-cyclic \cite{Rowell_09}, there are non-trivial modular non-abelien models with a larger number of charges which are non-cyclic. Examples include $D \left( \bar{D}_{2} \right)$, the quantum double constructed from the quaternion group, which contains $22$ topological sectors \cite{Bais_92, Propitius_96}. This model is thought to describe certain phases in a simple spin-$1$ SU(2) quantum magnet \cite{Xu_12}. There also exist families of non-abelian anyons which are non-cyclic by construction \cite{Tambara_98}.

\section{Errors}
\label{sec_noise}

In this section we describe the noise model and prove that errors can be organized in a hierarchical manner suitable for a renormalization-type analysis. 

\subsection{The Noise Model}

We consider here a surface $\Sigma$ of genus $g = 1$ discretized in a square lattice $\Lambda$ of size $Q^{n} \times Q^{n}$ with the identification of opposite boundaries and where $Q$ and $n$ are both positive integers. We assume that the excitations of the system are described by an anyonic model denoted by $\mathcal{A}$. A Wilson loop is present at each lattice site, allowing for local projective charge measurement. Time is also discretized. At every time step, charge noise is applied. In order to do so, all the edges of the graph are sequentially selected in a random order, the order changing at every time step. For each edge $( \mathbf{r,r'} )$, a single error process is selected among a set $ \mathcal{P} = \{ e_{c,c'} \}$ with probability $P(e_{c,c'})$, where $c$ and $c'$ label the topological charges (including the vacuum) of the anyons to be added onto sites $\mathbf r$ and $\mathbf r'$. To properly model local noise processes, we impose that $c' = \bar{c}$ and that their global charge is the vacuum.  We define the particle creation rate $\displaystyle p = \sum_{c,c'  \neq 1} P \left( e_{c,c'} \right)$. 

Using this noise model, it is possible that many anyons are located on the same site, since a site is connected to four edges. The fusion state of all the anyons on a site is in general described by a superposition of eigenstates of their total charge, since various fusion channels are available to non-abelian anyons. We assume that at time $t+\frac 12$, the excitations located on a given site instantaneously fuse into a superselection sector with appropriate fusion rules and probabilities. This can be viewed as an extra source of decoherence where the environment monitors the topological charge at each site. Moreover, as discussed in the introduction, the fusion process introduces some intrinsic irreversibility in the noise model, which is a significant departure from abelian models.

After having created thermal excitations and fused the excitations located at a single site into a definite superselection sector, we perform a noisy topological charge measurement at every site.  For every site, with probability $\left( 1-q \right)$ the right charge is reported and with probability $q$, a measurement error happens. In the case where there is a measurement error, one of the wrong charges $\{ c_{i} \}$ with $0 < i < \vert \mathcal{A} \vert$ is reported with probability of $q_{i}$. Here, $\vert \mathcal{A} \vert$ denotes the number of topological sectors of $\mathcal{A}$. The values of $q_{i}$ are subjected to the constraint that $\sum_{i} q_{i} = 1$. 

Another indirect source of errors stems from the error-correction procedure itself. Indeed, based on the reported charge measurements, decisions will be taken to move some excitations from one site of the lattice to another, following some transition rules described in Appendix \ref{AppendixA}. If a measurement reports the wrong charge, the wrong transition rule might be applied. In the case where a transition rule attempts to move an anyon of type $a$ from a site which does not contain such a charge, the effect of applying the transition rule is to create (from the vacuum) a pair of anyons of type $a$ and $\bar{a}$, to put the $\bar{a}$ charge in the site where the wrong measurement took place, while the $a$ anyon will be displaced to the site where the transition rule was attempting to displace the anyon. In this model, faulty measurements are effectively converted into charge errors. 

In what follows, let $\mathcal E$ be the set of all errors in the system, both charge errors and measurement errors. The $t^{\rm th}$ time step is thus decomposed into four substeps. Excitations are created between time $t$ and $t+\frac 12$. At time $t+\frac 12$, the charge at every site is collapsed onto a superselection sector. Immediately after this collapse, the topological charge at every site is measured in a faulty manner; the outcome of the measurement is called the error syndrome. The transition rules are executed between time $t+\frac 12$ and $t+1$. Since the transition rules can also result in multiple particles on a single site, the time step ends with a collapsed onto a superselection sector on every site at time $t+1$.

\subsection{Noise Classification}
\label{subsec_noise_class}

Given this noise model, error events can be classified in a hierarchic way as was shown by Harrington \cite{Harrington_04}, generalizing the one-dimensional analysis of G\'acs \cite{Gacs_86,Gray_01}. The definitions and results of subsections \ref{subsec_noise_class} and \ref{bound_error_rate} are taken from \cite{Harrington_04}, and are included here for self-containment.
In what follows, $Q$ and $U$ are chosen such that $Q \geq 4 \left( a+2 \right)$ and $U \geq 4\left( b+2 \right)$, where $a$ and $b$ are positive integers greater than $1$.

%In an abuse of notation, we identify an error $e$ acting on the edge of the lattice described by the coordinates $\left( x, y, t \right)$ with the edge $E \left( x, y, t \right)$ itself, where $x$ and $y$ denote the spatial position of the midpoint of the edge connecting the two sites on which $e$ acts, and $t$ is the midpoint of the time coordinate of the edge.

An error $e$ is defined on an edge of the lattice, which is characterized by a point in space-time denoted by $p \left( e \right) = \left( \mathbf r, t \right)$, where $\mathbf r = (x,y)$ denotes the spatial position of the midpoint of the edge connecting the two sites on which $e$ acts, and $t$ is the midpoint of the time coordinate of the edge.

Two sets of points $A$ and $B$ are said to be $\left( l, m, n \right)$-linked if there exists a space-time box of coordinates $\left[ x, x+l \right) \times \left[ y, y+m\right) \times \left[ t, t+n \right)$ containing at least one element in $A$ and one element in $ B$. If $A$ and $B$ are not $\left( l, m, n \right)$-linked, then they are said to be $\left( l, m, n \right)$-separated.

\subsubsection{Level-$0$ Noise}

The set of errors $S \subseteq \mathcal E$ ($S \neq \emptyset$) is a level-$0$ error candidate if it does not contain a mixture of charge errors and measurement errors, and if $S$ fits in a space-time box of size $\left[ x, x+1 \right] \times \left[ y, y+1 \right] \times \left[t, t\right]$. Additionally, $S$ is an actual level-$0$ error if it is $\left( a, a, b \right)$-separated from $\mathcal E \backslash S$. The union of all actual level-$0$ errors is called level-$0$ noise and is denoted by $E_{0}$.

\subsubsection{Level-$n$ Noise}
\label{sec_lvl_n_noise}
Level-$n$ candidate and actual errors are defined inductively. Suppose that level-$k$ candidate errors, actual errors, and level-$k$ noise $E_k$ are well-defined for any $k < n$. A non-empty set $S \subseteq \mathcal{E} \backslash E_{n-1}$ is a candidate level-$n$ error if 
\begin{itemize}[leftmargin=1.cm]
\item[(i)] $S$ is contained within a box of size $Q^{n} \times Q^{n} \times U^{n}$, and
\item[(ii)] {$S$ contains at least 2 disjoint candidate level-$\left( n-1 \right)$ errors that are \newline
$\left(aQ^{n-1},aQ^{n-1},bU^{n-1}\right)$-linked.}
\end{itemize}
$S$ is an actual level-$n$ error if additionally:
\begin{itemize}[leftmargin=1.cm]
\item[(iii)]  {$S$ does not contain two candidate level-$n$ errors that are \newline
$\left(4\left(a+2\right)Q^{n-1}, 4\left(a+2\right)Q^{n-1},4\left(b+2\right)U^{n-1}\right)$-separated, and}
\item[(iv)]  $S$ and $\mathcal{E} \backslash \left( S \cup E_{n-1} \right)$ are $\left( aQ^{n},aQ^{n},bU^{n}\right)$-separated.
\end{itemize}
Level-$n$ noise $E_{n}$ is defined as the union of the level-$n$ actual errors and of $E_{n-1}$.
Giving these definitions, one can show that a level-$n$ actual error always fits in a box of size $min \{ Q, 7\left( a + 2 \right) \}Q^{n-1} \times min \{ Q, 7\left( a + 2 \right) \}Q^{n-1} \times min \{ U, 7\left( b + 2 \right) \}U^{n-1}$ \cite{Harrington_04}.

\subsection{The renormalized error rate $\epsilon_{n}$}
\label{bound_error_rate}

The level-$n$ error rate $\epsilon_{n}$ is defined to be the probability that a box of size $Q^{n} \times Q^{n} \times U^{n}$ has non-empty intersection with at least one candidate level-$n$ error.

\begin{lemma}
The level-$n$ error rate $\epsilon_{n}$ is upper bounded by $\left( 4Q^{4} U^{2} \left( p + q \right)  \right)^{2^{n}}$.
\label{lem_error_rate}
\end{lemma}

\begin{proof}
Observe that $\epsilon_{n}$ is bounded above by the probability that at least one candidate level-$\left( n-1 \right)$ has non-empty intersection with a box of size  $Q^{n} \times Q^{n} \times U^{n}$, by definition of level-$n$ candidate errors. By union bound, this probability is bounded by  $\left( Q^{2}U \right) \epsilon_{n-1}$. Since each level-$n$ candidate error is composed of at least 2 level-$\left(n-1\right)$ candidate errors, we conclude that $\epsilon_{n} \leq \left( Q^{2}U\epsilon_{n-1} \right) ^{2}$. Now, we can use this recursive equation and the fact that the level-$0$ error rate is bounded above by $\epsilon_{0} \leq 4 \left( p+q \right)$. Given these considerations, we find that
\begin{eqnarray}
	\epsilon_{n} \leq \left( Q^{4}U^{2}4\left(p+q \right) \right)^{2^{n}}.
	\label{error_bound}
\end{eqnarray} 
\end{proof}

\begin{lemma}
Let $p$ and $q$ be such that $p + q < \frac{Q^{-4} U^{-2}}{4}$. Then, any error $e$ is part of an actual level-$n$ error, for some finite $n \in \mathbb{N}$, with probability 1.
\label{lem_actual_error}
\end{lemma}

\begin{proof}
Let $\mathpzc{E}_{n}$ be the event that there exist a box of size $Q^{n} \times Q^{n} \times U^{n}$ having non-empty intersection with at least one candidate level-$n$ error. Lemma~\ref{lem_error_rate} shown that $P \left( \mathpzc{E}_{n} \right) = \epsilon_{n} \leq \alpha^{2^{n}}$, for some positive constant $\alpha < 1$, as long as  $\left(p+q\right) < \frac{Q^{-4}U^{-2}}{4}$. $\sum^{\infty}_{i = 0} P \left( \mathpzc{E}_{i} \right)$ thus converges.

Suppose that there exists an error $e$ which is not part of any actual error. By definition, $e$ is a level-$0$ candidate error. Since it is not an actual level-$0$ error, $e$ is not $\left( a, a, b \right)$-seperated from $\mathcal E \backslash \{e\}$ and there exists $\{ s \} \subset \mathcal E$ which is $\left( a, a, b \right)$-linked with $\{ e \}$. Consider next $S = \{ e \} \cup \{s \}$. By definition, it is included in a box of size $Q \times Q \times U$, and contains 2 disjoint level-$0$ error candidates who are $\left( a, a, b \right)$-linked. $S$ is thus a level-$1$ error candidate. Since ${e}$ is not part of a level-$1$ actual error, condition (iv) must be violated, since $S$ fulfils condition (iii) by construction. Thus, there exist $S'$ disjoint from $S$ such that it is $\left( aQ, aQ, bU \right)$-linked with $S$, and $S'$ can be chosen to be small enough to be a level-$1$ candidate error. By construction, $S \cup S'$ fits in a box of size $\left( a+2 \right) Q \times \left( a+2 \right) Q \times \left( b + 2 \right) U$, which is smaller than $ Q^{2} \times Q^{2} \times U^{2}$, and contains 2 disjoint level-$1$ error candidates who are $\left( aQ, aQ, bU \right)$-linked. $S \cup S'$ is then a level-$2$ candidate error. At this point, the same reasoning can simply be repeated. We then see that $\{ e \}$ is part of a candidate error of level-$n$, for any value of $n \in \mathbb{N}$. The Borel-Cantelli lemma tells us that such an error has a probability of 0 of happening.\footnote{The Borel-Cantelli lemma~\cite{Borel_1909,Cantelli_1917} shows that if a sequence of events $\mathpzc{E}_{n}$ is such that $\sum_{n=0}^{\infty} P(\mathpzc{E}_{n})$ converges, then $P \left( \cap_{n=0}^{\infty} \cup_{k=n}^{\infty} \mathpzc{E}_{k} \right) = 0$.}
\end{proof}

\section{The error-correction procedure}
\label{sec_corr_algo}

Let the variable $c \left( \bm r, t \right)$ describes the true topological charge present at site $\bm r$ where $t$ can take both integer and half-integer values, the latter describing the charge in a site after the measurement process, and the former describing the charge after the application of transition rules. The syndrome $s_{0} \left( \bm r , t + \frac{1}{2} \right)$, with $t$ an integer, describes the reported topological charge after the measurement performed at time $t + \frac{1}{2}$ , which differs from $c\left( \bm r, t+ \frac{1}{2} \right)$ with probability at most $q$. The subscript $0$ refers to the fact that $s_{0} \left( \bm r, t + \frac{1}{2} \right)$ describes the syndrome of a physical site, which corresponds to the $0^{\rm th}$ level of the renormalization scheme used by the cellular automaton decoder. Both $c$ and $s_{0} $ can take $\vert \mathcal{A} \vert$ unique different values, representing the various admissible topological charges of the model. 

The lattice is separated into square colonies of size $Q \times Q$, with $Q$ odd for convenience. Transition rules are applied at every time steps conditioned on the syndrome.  A full description of the transition rules is quite lengthy, and can be found in appendix \ref{AppendixA}, together with a description of the structure of a colony. The transition rules applied at a given site depend both on the position of the site in the colony and on the syndrome at the site under consideration and on its 8 nearest neighbours. The idea behind the transition rules at the physical level is essentially to bring non-trivial charges that were created locally together. If a non-trivial charge is detected on a site, then 2 cases are possible: if another non-trivial charge is detected in a neighbouring site, the transition rules are such that the charges will be brought together, and hopefully fuse to the vacuum. In the case where no neighbouring charges are detected, the rules will make the charge move towards the colony centre, where hopefully the charges stemming from a local error will all fuse together into the vacuum.

To describe the level-$1$ transition rules, we coarse grain time into $b$ bins of duration $b$.  At every $U = b^{2}$ time steps ($b$ bins of $b$ time steps each), the colony centres are organized into ``renormalized colonies" (which we refer to as a level-$1$ colony) of size $Q^{2} \times Q^{2}$, as depicted in Figure~\ref{renorm_array}. For each colony, two other syndrome measurements are defined, $s_{1,c} \left( \bm \rho_1, \tau_1 \right)$ and $s_{1,n} \left( \bm \rho_1, \tau_1 \right)$, where $(\bm \rho_1,\tau_1)$ are the space-time coordinate of the renormalized lattice. The $\tau_1^{\rm th}$ working period extending from time step $t=\tau_1 U$ to $(\tau_1+1)U-1$. For each time bin, a topological charge is declared to be present at the colony centre if for at least $f_{c}b$ time steps of the bin a non-trivial charge was reported at the colony centre. Furthermore, in the case where a non-trivial charge is reported, the reported charge is the one which was reported most often during the time bin. If a fraction at least $f_{c}$ of the $b$ time bins have a non-zero syndromes, then $s_{1,c} \left( \bm \rho_1 , \tau_1 \right)$ is set to the value corresponding to the last non-trivial charge reported during a bin of $b$ time steps. The same procedure is used to determine the value of $s_{1,n} \left( \bm \rho_1 , \tau_1 \right)$, with the difference that the fraction $f_{n}$ is used.

The application of the transition rules at this renormalized level is similar to the physical ones, using $s_{1,c}$ for the syndrome measurements at the centre position and $s_{1,n}$ as syndrome measurements for the eight neighbouring positions. The motivation for this renormalization procedure is to deal with errors that spatially extend into more than a single colony. The principle is that the charges present in every colony will eventually be moved to their respective centres, where they will either fuse to the vacuum or eventually activate the corresponding level-1 syndromes. At this point, the non-trivial charges present in the colony centres will be displaced towards the level-$2$ colony (formed of $Q \times Q$ level-$1$ colonies) centre containing them, where all the non-trivial charges will eventually merge and fuse to the vacuum, assuming that the error fits inside a single level-$2$ colony.

\begin{figure}[h]
\centering
\includegraphics[scale=.3]{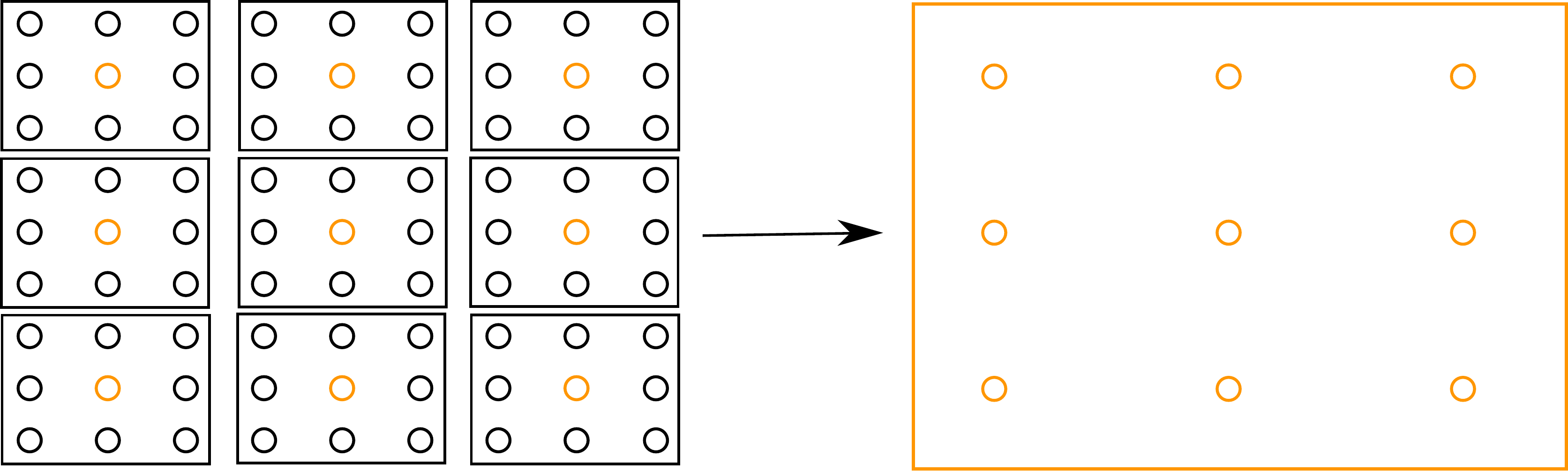}
\caption{A set of $Q \times Q$ colonies are renormalized to form a single super-colony, where the old colony centres form the sites of the renormalized colony. Here $Q = 3$.}
\label{renorm_array}
\end{figure}

This procedure is repeated for higher levels of renormalizations in a self-similar fashion in order to deal with larger errors. At every $U^{k}$ time steps (a level-$k$ working period), the $k^{\rm th}$ renormalization step is applied, by forming level-$k$ colonies of size $Q^{k+1} \times Q^{k+1}$ and by using the centres of the renormalized colonies of the previous level of size $Q^{k} \times Q^{k}$ in order to determine the syndromes $s_{k,c} $ and $s_{k,n} $. A level-$k$ colony can be labelled by $C^{k}_{\bm \rho_k}$, where $\bm \rho_k$ refers to the position of the level-$k$ colony in the level-$k$ renormalized lattice, and a level-$k$ working period covering the time period $\left[ \right. \tau_{k} U^{k}, \left( \tau_{k}+1 \right) U^{k} \left. \right)$ can be labelled by $\tau_k$. The same binning procedure is used to determine the level-$k$ syndromes $s_{k,c} (\bm \rho_k, \tau_k)$ and $s_{k,n} (\bm \rho_k, \tau_k)$, using the same thresholds $f_{c}$ and $f_{n}$. The transition rules are then applied at the renormalized level, depending on the values of the  level-$k$ syndromes.

If $\bm r$ is at the centre of the level-$k$ colony $C^k_{\bm \rho_k}$, we refer to $\bm \rho_k$ as $\bm \rho_k (\bm r)$ when considering the level-$k$ colony in relation with the physical site at its centre. For the remainder of the text, we do not explicitly specify the level of renormalization for $\bm \rho$ and $\tau$ when it is clear from the context. We also use $\tau$ to label the time index of a working period, while $t$ refers to the actual (or level-$0$) time. The same goes for $\bm \rho$ representing the position of a colony in the renormalized lattice and $\bm r$ representing a physical site. 

The procedure can be summarized by the following pseudo algorithm, in which $t_{fin}$ is the number of time steps for the algorithm to run and $n_{max} = log_{Q} \left( L \right)$ with $L$ the linear size of the system :

\begin{algorithm}[!h]
\begin{algorithmic}[0]
\caption{The fault-tolerant correcting algorithm}\label{algo}
\State t = 0;
\While{$t < t_{fin}$}
	\State apply charge noise;
	\State $t \gets t + \frac{1}{2}$
	\State measure (in a faulty manner) the charges at every site;
	\For{$n = 0; n < n_{max}; n= n+1$}
		\If {$\left( t+\frac{1}{2} \right)$ modulo $U^{n}$ = $0$}
			\State apply level-$n$ transition rules based on reported syndromes;
		\EndIf
	\EndFor
	\State $t \gets t + \frac{1}{2}$
\EndWhile
\end{algorithmic}
\end{algorithm}

For convenience, we list here the various constants used by the decoding algorithm and the noise model and summarize their purpose.
\begin{itemize}
\item $p = $ rate of particle creation.
\item $q = $ rate of measurement errors.
\item $D = diam \left( G_{\mathcal{A}} \right)$ is the maximum number of fusions between a particle and its anti-particle required before the fusion outcome is the vacuum.
\item $a,b\geq 2$ are two arbitrary constants.
\item $Q \geq 4(a+2)$ is the linear dimension of a colony. A level-$k$ colony has dimension $Q^k\times Q^k$ and is formed of $Q\times Q$ level-$(k-1)$ colonies.
\item $U = b^2 \geq 4(b+2)$ is the duration of a working period. A level-$k$ working period has length $U^k$ and is made up of $b$ bins, each composed of $b$ level-$(k-1)$ working periods. 
\item $f_c$ is the fraction of positive syndrome measurements required to consider applying the local transition rule in a colony. We will later set $f_c  = (b-4(3D + 1)Q - 1) / b$ to ensure that the correction procedure is correct.
\item $f_n$ is the fraction of positive syndrome measurements required to consider that a neighbouring colony should be part of a local transition rule. We will later set $f_n  = ( 4(3D + 1)Q + 1 ) / b$ to ensure that the correction procedure is correct.
\end{itemize}

We note that these cellular automaton transition rules are not strictly local because level $k>0$ transition rules involve physical sites that are $Q^k$ cells apart. G\'acs deploys considerable efforts to implement these high-level rules in a strictly local fashion, and generalized these results for a two-dimensional array of self-simulating arrays of cellular automata \cite{Gacs_89}. In addition to implementing $k$-level updates rules, the $k$-level cells implement additional local operations in order to realize the $(k+1)$-level rules, and so forth. While we have no reasons to doubt that such a strictly local set of rules is possible in the current setting, we do not attempt to develop them here.

\section{Definition of Basic Concepts for Error-Correction of Non-Abelian Anyons}
\label{sec:defs}

G\'acs' proof that one-dimensional cellular automatons can process information in a fault-tolerant way is notoriously complex. Harrington builds on this proof by explaining how key concepts need to be adapted to the toric code setting. Similarly, our proof builds on Harrington's proof and focuses on the key novelties introduced by the non-abelian nature of the anyons. The goal of this section is to introduce some basic concepts required to analyse error-correction of non-abelian anyons and to present some of their key properties. 

The core idea in G\'acs and Harrington's approach is to classify errors into distinct levels as in Section \ref{subsec_noise_class} and to demonstrate that level-$k$ errors are effectively suppressed by level-$k$ transition rules. In a sense, errors of different levels do not interact with each other and can be analyzed independently. The non-trivial braiding relations and fusion rules of non-abelian anyons break this simple structure. 

The fusion rules of non-abelian anyons are in general non-deterministic. As in Harrington's approach, we will show that the level-$k$ syndrome will correctly identify the topological charge of a cell. But even when a particle and its anti-particle have been correctly identified, their fusion may result in a non-trivial particle. As a consequence, our proof needs to apply to all possible fusion histories of the anyons. For this purpose, we will introduce the notion of the trajectory domain of an error, which is roughly the set of sites that have a finite amplitude of becoming charged as a consequence of a given error, and where transition rules are taken into account. 

The non-abelian braiding relations have a deeper consequences on the error classification. Consider a situation where a low-level actual error $ E$ is well isolated, in space and time, from any other actual error. The low-level correction rules tend to concentrate all the excitations caused by $ E$ onto a single site, which will result in the vacuum and thus the elimination $ E$. Suppose, however, that a high-level transition rule drags an anyon through the region containing $ E$. This is not forbidden by the definition of level-$k$ errors: those need to be well isolated from each other and from higher-level errors, but this does not prevent high-level transition rules from operating in their vicinity. As a consequence of the non-abelian braiding rules, after the passage of the high-level anyon, the excitations created by $ E$ may no-longer fuse to the vacuum: they have become entangled with the high-level error. Thus, neither the error $ E$ nor the higher-level error with which it has become entangled can be corrected individually: they need a joint correction strategy.

This entanglement across errors of different levels requires the definition of causally-linked clusters of errors. These are collections of actual errors of distinct levels that have potentially become entangled through the transition rules applied by the correction algorithm. Despite this new failure mechanism, repeated applications of the correction rules are bound to succeed for non-cyclic anyon models. Indeed, every failed attempt moves the total charge of the anyon which is dragged by the high-level transition rule along the directed acyclic graph defined in Section \ref{sec_cyclic_anyons}. After at most $diam \left( G_{\mathcal{A}} \right) - 1$ attempts, that topological charge will become abelian, so the next iteration is guaranteed to succeed as in an abelian model. Thus, the net effect is a possible slow-down of the correction process, which can be compensated by a lower error threshold. 

With these motivations in mind, we now give some basic definitions required for the analysis of our decoding algorithm. Henceforth, we assume that $p+q < \frac{Q^{-4}U^{-2}}{4}$, so that every error is part of an actual error of a finite level, and the notion of level-$k$ noise is well defined, $\forall \textnormal{ } k \in \mathbb{N}$.

\subsection{State Evolution}

Let $\vert \psi_{0} \rangle$ be a state of the system at time $t = 0$, part of the ground space, and let $\mathcal{E} = \{ E(t) \}$ with $t \geq 0$ be the set of all errors affecting the system, such that $E(t)$ are the errors affecting the system at time $t$. Note that $E(t)$ is not in general a set of actual errors. 

The evolution of the system from integer time $t$ to time $t + 1$ given a specific set of fusion outcomes is decomposed in two elementary steps. The first one describes the application of the charge errors that have time coordinate $t$ and the subsequent fusion of the topological charge with outcomes $c ( \bm r , t + \frac{1}{2} )$ at every site $\bm r$ of the lattice. The second step describes the application of the transition rules based on the reported syndromes $s_{ 0 } \left( \bm r , t+\frac{1}{2} \right)$, followed by the fusion of the charges contained at every site to $c ( \bm r, t + 1 )$. The charge errors described by  integer time coordinate $t$ are applied during the time interval $\left( t, t+\frac{1}{2} \right)$, so that the state of the system at time $t$ does not incorporate them. Similarly, the syndrome measurements and application of the transition rules take place during the time interval $\left[ t+\frac{1}{2}, t+1 \right)$.

For any  positive integer  time $t \geq 0$, the state of the system is given by $\vert \psi \left( c \left( t \right), \mathcal{E} \right) \rangle$, where $c \left( t \right)$ is the set of measured topological charges (not necessarily the reported ones) at every site of the lattice $\Lambda$ up to time $t$, and $\mathcal{E}$ is the set of all errors having affected the system. $s(t)$ is defined similarly for the set of all level-$0$ syndromes at times prior or equal to $t$. Implicit in this notation is the fact that the state $\vert \psi \left( c \left( t \right),  \mathcal{E} \right) \rangle$ depends on the transition rules that have been applied to the system, but these are entirely determined by $s(t)$, which itself depends on the fusion outcomes $c(t)$ and the measurement errors, specified by $\mathcal E$. 

The equation describing the first process is
\begin{eqnarray} 
\vert \psi  \left( {c} \left( t + \tfrac{1}{2} \right),  \mathcal{E} \right) \rangle = \frac{\prod_{\mathbf{r} \in \Lambda} P^{c \left( \mathbf{r},  t + \frac{1}{2} \right)}_{ \mathbf{r} } \hat{E} \left(t \right) \vert \psi \left( c \left( t \right),  \mathcal{E} \right) \rangle }{\sqrt{P \left( c \left( t+ \frac{1}{2} \right) \vert c \left( t \right),  \mathcal{E}  \right)}},
\label{eqn_evolution_half}
\end{eqnarray}
\begin{eqnarray}
P \left( c \left( t + \tfrac{1}{2} \right) \vert c \left( t \right),   \mathcal{E}  \right) = \langle \psi \left( c \left( t \right),   \mathcal{E} \right) \vert \hat{E} \left( t \right) ^{\dagger} \prod_{\mathbf{r} \in \Lambda} P^{c \left( \mathbf{r},  t+ \frac{1}{2} \right)}_{ \mathbf{r} } \hat{E} \left( t \right) \vert \psi \left( c \left( t \right), \mathcal{E} \right) \rangle, \end{eqnarray}%
where $P^{c \left( \mathbf{r}, t \right)}_{\mathbf{r}}$ is the projector of the anyons present in site $\mathbf{r}$ at time $t$ onto the total charge $c \left( \mathbf{r}, t \right)$ and $\hat{E} \left( t \right)$ is the operator applying the errors in $E \left( t \right)$. The probability of a charge measurement outcome $P \left( c \left( t + \tfrac{1}{2} \right) \vert c \left( t \right),   \mathcal{E}  \right)$ is simply given by Born's rule.\footnote{The probability of a fusion outcome can in some cases depend on the ground state in which the system was prepared. This happens when the errors+corrections have build up to a logical error. This dependence is not explicitly included by the notation used in Eq.~\ref{eqn_born_rule} which, to avoid cluttering the notation, shows only the dependence on the previous fusion outcomes and the error history.  We can circumvent this issue more formally by assuming that the initial state of the system is in the ground subspace, but otherwise maximally mixed, or equivalently maximally entangled with a reference system. }

The second process is described by

\begin{equation}
\vert \psi \left( c \left( t+1 \right), \mathcal{E} \right) \rangle = \frac{\prod_{\mathbf{r} \in \Lambda} P_{\mathbf{r}}^{c \left( \mathbf{r}, t+1 \right)} \hat{T} \left( {s} \left( t + \frac{1}{2} \right) \right)  \vert \psi \left( c \left( t + \frac{1}{2} \right),   \mathcal{E} \right) \rangle } {\sqrt{P \left( c \left( t+1 \right) \vert c \left( t+ \frac{1}{2} \right),  \mathcal{E}  \right)}},
\label{eqn_evolution_one}
\end{equation}
\begin{equation}
\begin{aligned}
P \left( c \left( t+1 \right) \vert c \left( t+\tfrac{1}{2} \right),  \mathcal{E}  \right)
 = \langle \psi \left( c \left( t+ \tfrac{1}{2} \right),  \mathcal{E} \right) \vert \hat{T}^{\dagger} & \left( s \left( t + \frac{1}{2} \right) \right) \prod_{\mathbf{r} \in \Lambda} P^{c \left( \mathbf{r},  t+1 \right)}_{ \mathbf{r} } \hat{T} \left( s \left( t + \frac{1}{2} \right) \right) \vert \psi \left( c \left( t+\tfrac{1}{2} \right), \mathcal{E} \right) \rangle,
\label{eqn_born_rule}
\end{aligned}
\end{equation}%
where $\hat{T} \left( s \left( t \right) \right)$ are the transition rules applied based on the faulty syndromes reported up to time $t $. Here we use the fact that level-$1$ and higher syndromes are entirely determined by level-$0$ syndromes, so that the transition rules that are applied can be considered to be functions of level-$0$ syndromes as well.

Given a pattern $\mathcal{E}$ of actual errors, it is not possible in general to deterministically predict the charge of a specific site due to the probabilistic nature of projecting the total charge of non-abelian anyons. However, it is possible to describe this charge as a random variable $C \left( \mathbf{r}, t \right)$ and the associated probability distribution $P \left( C \left( \mathbf{r}, t \right) , \mathcal{E} \right)$ can be computed using the probability chain rule 

\begin{equation}
P \left( C \left( \mathbf{r}, t \right) = c_{0}, \mathcal{E} \right) = \sum_{c \left( t \right) \vert c \left( \mathbf{r}, t \right) = c_{0}} \prod_{t' = 0}^{t-\frac{1}{2}} P \left( c \left( t' + \tfrac{1}{2} \right) \vert c \left( t' \right) , \mathcal{E} \right).
\label{eqn_prob_s}
\end{equation}.

\subsection{Causal regions}

We are now in a position to define the notion of a causal region of an actual error. As mentioned previously, the noise affecting the system $\mathcal{E}$ can be decomposed as a collection of actual errors. In the remainder of this section, we work with the convention that $\mathcal{E}$ is the collection of all the actual errors, and therefore $E \in \mathcal{E}$ means that $E$ is an actual error in the set $\mathcal{E}$. When the context is clear, we will also view $E$ as a set of points in space-time, those at which errors in $E$ are present. 

Errors are physical processes which cause anyonic excitations to appear in the system. These excitations are imperfectly monitored and moved following predefined transition rules. The transition rules do not act directly on the  errors, but rather on the anyonic excitations. Intuitively, different actual errors are treated independently, as they are isolated from each other. However, it is possible for anyons created by different actual errors to interact together {\it via} braiding or fusion processes, or by modifying the application of the transition rules. The fact that the anyons are non-abelian makes it impossible to treat actual errors independently. The following definitions allow us to treat all the anyons which may come from different actual errors but that are nevertheless interacting together as a single entity.

\begin{definition}
For $E \in \mathcal{E}$ and $\{ E_{(i)} \} \subseteq \mathcal{E}$, let $\mathcal Z \left( \{ E_{(i)} \} ,E ; t \right)$ be the set of points in space time $p = \left( \mathbf{r}, t_{0} \right)$ with $t_{0} \leq t$ such that $P \left( C \left( \mathbf{r}, t_{0} \right), \{ E_{(i)} \} \right) \neq P \left( C \left( \mathbf{r}, t_{0} \right), \{ E_{(i)} \} \backslash \{ E \} \right)$. The causal region of $E$ at time $t$ is defined as $\displaystyle Z \left( E ; t \right) = \bigcup_{\{ E_{(i)} \} \subseteq \mathcal E } \mathcal Z \left( \{ E_{(i)} \}, E ; t \right) $, where the union is performed over all possible subsets $\{ E_{(i)} \}$ of $\mathcal E$.
\end{definition}

% WITH DEFINITION 1
%\begin{definition}
%For $E \in \mathcal{E}$, let $Z \left( E ; t \right)$ be the set of points in space time $p = \left( \mathbf{r}, t_{0} \right)$ with $t_{0} \leq t$ such that $P \left( C \left( \mathbf{r}, t_{0} \right), \mathcal{E} \right) \neq P %\left( C \left( \mathbf{r}, t_{0} \right), \mathcal{E} \backslash \{ E \} \right)$. $Z \left( E ; t \right)$ is called the causal region of $E$ at time $t$. 
%\end{definition}

If $Z \left( E; t \right) \cap Z \left( E' ; t \right)\neq \emptyset$, then we say that at time $t$, $E$ and $E'$ have interacted together. The set of actual errors that have interacted, either directly or indirectly with $E$ is called the causally linked cluster of $E$, formally defined as follows.
\begin{definition}
For $E \in \mathcal{E}$, let $\mathcal{C}_{ 0 } \left( E ; t \right) = \{E\}$. $\forall i \in \mathbb{N}$, define $\mathcal{C}_{ i + 1 } \left( E ; t \right) = \{ E' \in \mathcal{E} \vert  Z \left( E^{\prime\prime} ; t \right) \bigcap Z \left( E' ; t \right) \neq \emptyset\ {\rm for}\ E^{\prime\prime} \in \mathcal C_i(E:t)   \}$. The causally-linked cluster of the actual error $E$ at time $t$ is $\mathcal{C} \left( E ; t \right) = \bigcup_{i} \mathcal{C}_{ i } \left( E ; t \right) $. When no time is specified, the causally-linked cluster of $E$, $\mathcal{C} \left( E \right)$ is defined as $\displaystyle \bigcup_{t} \mathcal{C} \left( E ; t \right) $.
\label{def_partial_c}
\end{definition}
Any actual error $E \in \mathcal{E}$ is in at least one causally-linked cluster, since $E \in \mathcal{C} \left( E ; t \right) $ for any $t$.
\begin{proposition}
At any time $t$, any actual error $E$ is contained within a single causally-linked cluster, which can be denoted by $\mathcal{C} \left( E ; t \right)$.
\label{prop_single_cluster} 
\end{proposition}

\begin{proof}
Suppose it is not the case, {\it i.e.} there exists two different causally-linked clusters $\mathcal{C} \left( E' ; t \right)$ and $\mathcal{C} \left( E^{\prime\prime} ; t \right)$ such that $\mathcal{C} \left(E' ; t \right) \neq \mathcal{C} \left( E^{\prime\prime} ; t \right)$ and such that $E \in \mathcal{C} \left( E' ; t \right)$ and $E \in \mathcal{C} \left( E^{\prime\prime} ; t \right)$. We thus have that $E \in \mathcal{C}_{ i } \left( E' ; t \right)$ and $E \in \mathcal{C}_{ k } \left( E^{\prime\prime} ; t \right)$, for $i, k \in \mathbb{N}$. But using definition \ref{def_partial_c}, we find that there exists $F \in  \mathcal{C}_{i+1} \left( E' ; t \right)$ such that it is also in $\mathcal{C}_{ k - 1 } \left( E^{\prime\prime} ; t \right)$. This is because if $E$ is contained in $\mathcal{C}_{ i } \left( E' ; t \right)$, there exists an actual error $F$ such that it interacts with $E$ and such that it is contained in $\mathcal{C}_{i-1} \left( E' ; t \right)$. Consequently, $F$ is contained in $\mathcal{C}_{k+1} \left( E^{\prime\prime} ; t \right)$. By repeating the same reasoning, we find that $F \in \mathcal{C}_{k+i} \left( E^{\prime\prime} ; t \right)$. Again, by definition \ref{def_partial_c}, we find that $\mathcal{C}_{ i } \left( E';t \right) \subset \mathcal{C}_{k+2i} \left( E^{\prime\prime} ; t \right)$. Thus, $\mathcal{C} \left( E' ; t \right) \subset \mathcal{C} \left( E^{\prime\prime} ; t \right)$. The above deductions are equally valid by interchanging the roles of $E'$ and $E^{\prime\prime}$, which show that $\mathcal{C} \left( E' ; t \right) \supset \mathcal{C} \left( E^{\prime\prime} ; t \right)$ as well, and from which we conclude that $\mathcal{C} \left( E' ; t \right) = \mathcal{C} \left( E^{\prime\prime} ; t \right)$, a contradiction.
\end{proof}

The above implies that if two actual errors $E'$ and $E^{\prime\prime}$ are in the same causally-linked cluster at time $t$, then $\mathcal{C} \left( E' ; t \right) = \mathcal{C} \left( E^{\prime\prime} ; t \right)$. We are thus free to choose any representative of the causally-linked cluster to identify it. A similar statement holds when the time is not specified.

We say that the causal region of  $\mathcal{C} \left( E ; t \right)$ is $\displaystyle \bigcup_{F \in \mathcal{C} \left( E ; t \right)} Z \left( F ; t \right)$. Depending on context, we will view $\mathcal{C} \left( E ; t \right)$ as a collection of errors or as a region of space-time, corresponding to the causal region of $\mathcal{C} \left( E ; t \right)$. Also, when a lattice site in $\mathcal{C} \left( E ; t \right)$ has a non-trivial charge, we will say that this anyon belongs to $\mathcal{C} \left( E ; t \right)$. Finally, we define the level of a causally-linked cluster at time $t$ by $ \displaystyle \textnormal{lvl} \left( \mathcal{C} \left( E ; t \right) \right) = \max_{E' \in \mathcal{C} \left( E ; t \right)} {\textnormal{lvl} \left( E' \right) } $, where $\textnormal{lvl} \left( E' \right)$ is the function returning the level of the actual error $E'$. Similar definitions holds when time is not specified.

Given these definitions, a point $p = \left( \mathbf{r}, t \right)$ cannot belong to more than a single causally-linked cluster. Suppose that $p$ belongs to $\mathcal{C} \left( E ; t \right)$ and to $\mathcal{C} \left( E' ; t \right)$. Using the freedom of representative error for a given causally-linked cluster, we can assume that $p \in Z \left( E ; t \right)$ and   $p \in Z \left( E' ; t \right)$ . Using definition \ref{def_partial_c}, it is clear that $E \in \mathcal{C} \left( E^{\prime} ; t \right)$. Using proposition \ref{prop_single_cluster}, we find that $\mathcal{C} \left( E; t \right) = \mathcal{C} \left( E^{\prime} ; t \right)$.

We also find that removing a causally-linked cluster $\mathcal C \left( E; t \right)$ does not affect points outside of its causal region. Consider a point $p = \left( \bm r, t_0 \right)$ such that $p \notin \mathcal C \left( E ; t \right)$. We thus have that $P \left( C ( \bm r , t_0 ), \{ E_{(i)} \} \right) = P \left( C ( \bm r, t_0 ), \{ E_{(i)} \} \backslash {F} \right)$ for any set of actual errors $\{ E_{(i)} \}$ and any actual error $F \in \mathcal C \left( E ; t \right)$. In particular, for any subset of actual errors $\mathcal{C}_{S} \subseteq \mathcal C \left( E ; t \right)$, we have that $P \left( C ( \bm r, t_0 ), ( \{ E_{(i)} \} \backslash \mathcal{C}_{S} ) \backslash F \right) = P \left( C ( \bm r, t_0 ), ( \{ E_{(i)} \} \backslash \mathcal{C}_{S} ) \right)$. We can thus remove all the actual errors of $\mathcal C \left( E ; t \right)$ to find that the probability distribution of $C \left( \bm r, t_0 \right)$ is  left unchanged. Causally-linked cluster can thus be analysed independently. Notice that if a point $p$ with time coordinte $t_0 \leq t$ is not in any causally-linked cluster at time $t$, then it must contain the vacuum.

Our proof of error-correction requires analyzing the effect of an actual error $E$ in the absence of errors of higher levels. This motivates the definition of a restricted error history $\mathcal{\tilde E} \left( E \right) = \{ E' \in \mathcal{E} \vert \textnormal{ lvl} \left( E' \right) < \textnormal{ lvl} \left( E \right)  \} \textnormal{ } \bigcup \textnormal{ } \{ E  \}$ and their associated  connected clusters.

\begin{definition}
For $E' \in \mathcal{\tilde E} \left( E \right)$ and $\{ E_{(i)} \} \subseteq \mathcal{\tilde E} \left( E \right)$, let $\mathcal{\tilde Z}_{E} \left( \{ E_{(i)} \}, E' ; t \right)$ be the set of points in space-time $p = \left( \mathbf{r}, t_{0} \right)$ with $t_{0} \leq t$ such that $P \left( C \left( \mathbf{r}, t_{0} \right), \{ E_{(i)} \} \right) \neq P \left( C \left( \mathbf{r}, t_{0} \right), \{ E_{(i)} \} \backslash E' \right)$. We call $\displaystyle \tilde{Z}_{E} \left( E' ; t \right) = \bigcup_{\{ E_{(i)} \} \subseteq \mathcal{\tilde E} \left( E \right) } \mathcal{\tilde Z}_{E} \left( \{ E_{(i)} \}, E' ; t \right)$ the restricted causal region of $E'$ with respect to $E$ at time $t$.
\end{definition}

We can use this notion of restricted causal region to define restricted causally-linked clusters $\mathcal{\tilde C}(E;t)$. These are defined as in Definition \ref{def_partial_c} but incorporate only restricted errors and causal regions, i.e.,  $\mathcal{\tilde C}_{ i + 1 } \left( E ; t \right) = \{ E' \in \mathcal{\tilde E}(E) \vert  Z_E \left( E^{\prime\prime} ; t \right) \allowbreak \bigcap Z_E \left( E' ; t \right) \neq \emptyset\ {\rm for}\ E^{\prime\prime} \in \mathcal{\tilde C}_i(E;t)   \}$. The notion of restricted clusters will become useful when analyzing the effects of actual errors of level lower than the cluster to which they belong. They do not form independent entities in general, since interactions with actual errors of larger level are not taken into account.

Lastly, our analysis of the hierarchical error-correction procedure will require a coarse-grained notion of the error clusters.

\begin{definition}
Let $\displaystyle \mathcal{C}^{\left( k \right)} \left( t \right) = \bigcup_{ E  \vert \text{lvl} \left( E \right) \geq k} \mathcal{C} \left( E ; t \right)$. 
\label{def:coarse_grained_cluster}
\end{definition}

\subsection{Trajectories}

The shape of a causally-linked cluster $\mathcal{C} \left( E ; t \right)$ will generally evolve in time, and to prove error-correction in the next Section, we need to bound this evolution. In order to do so, we introduce the notion of the trajectory of an error cluster $\mathcal{T} \left( \mathcal{C} \left( E ; t \right) \vert c(t), \mathcal E \right)$, which is the set of points in space-time up to time $t$ containing non-trivial anyons originating from errors in $\mathcal{C} \left( E ; t \right)$. Clearly, this depends on the fusion outcomes and the transition rules, and this is why $\mathcal T$ is explicitly a function of the fusion outcome and the error history (which, combined with the fusion outcomes, determine the syndromes and thus the transition rules).

Let $e$ be an `atomic' charge (measurement) error which happens at time $t_0$ ($t_0 + \frac{1}{2}$) and is contained in one of the actual errors of a level-$k$ causally-linked cluster. For $t$ such that $t_{0} \leq t \leq t_{f}$, the locations $\mathcal{L}_{t} \left( e \vert c(t), \mathcal E\right)$ is a set of points with time coordinate $t$ constructed by induction as follows. If $e$ is a charge error, $\mathcal L_{t_0}(e\vert c(t_0), \mathcal E)$ is composed of the two space-like separated sites adjacent to the error $e$, which contain a particle anti-particle pair. If $e$ is a measurement error, then $\mathcal L_{t_0}(e\vert c(t_0), \mathcal E)$ contains a single site where the measurement error occurred.

Suppose that $\mathcal{L}_{t'} \left( e\vert c(t'), \mathcal E \right)$ is well defined for any $t'$ such that $t_{0} \leq t' \leq t$. Then, the point $\left( i, j, t+1 \right)$ is in $\mathcal{L}_{t+1} \left( e \vert c(t+1), \mathcal E\right)$ if $c \left( i , j, t+1 \right) \neq 1$ and at least one of the following conditions is satisfied
\begin{enumerate}
\item $\left( i , j, t \right) \in \mathcal{L}_{t} \left( e\vert c(t), \mathcal E \right)$; 
\item $\left( i, j, t+1 \right) = \left( i_{0}+\alpha Q^{k'}, j_{0}+\beta Q^{k'}, t+1 \right)$, where $ \left( i_{0}, j_{0},  t' \right) \in \mathcal{L}_{t'} \left( e\vert c(t'), \mathcal E \right)$ is at the centre of a level-$k'$ colony and $M^{\left( \alpha, \beta \right), k'}_{\bm \rho (i_{0}, j_{0})}$ is applied at time $t + \frac{1}{2}$, with $k' \leq k$ and $t - t' < U^{k'}$; or
\item $\left( i, j, t' \right) \in \mathcal{L}_{t'} \left( e \vert c(t'), \mathcal E \right)$ is at the centre of a level-$k'$ colony and some transition rule $M^{\left( \alpha, \beta \right), k'}_{\bm \rho (i,j)}$ is applied at time $t + \frac{1}{2}$ with $k' \leq k$ and $t-t' < U^{k'}$.
\end{enumerate}
The first of these conditions corresponds to an anyon which stays put. The second one represents an anyon which moves due to a level-$k'$ transition rule. The last one represents a level-$k'$ transition rule which has left, or created, a non-trivial charge behind.

The trajectory of the error $e$ at time $t_f$ is simply the set of locations it has occupied up to time $t_f$: $ \mathcal{T}_{t_f} \left( e\vert c(t_f), \mathcal E \right) = \bigcup_{t \in [t_{0}, t_{0} + 1, \dots, t_{f} ]} \mathcal{L}_{t} \left( e\vert c(t), \mathcal E \right)$. We can define the trajectory (and similarly for locations) of a set of actual errors $\{ E_{\left( i \right) } \}$ by  $ \mathcal{T}_{t} \left( \{ E_{\left( i \right) } \}\vert c(t), \mathcal E \right) = \bigcup_{e \in E \vert E \in \{ E_{\left( i \right) } \}  } \mathcal{T}_{t} \left( e \vert c(t), \mathcal E\right)$.

For a level-$k$ causally-linked cluster $ \mathcal{C} \left( E ; t \right)$, in the absence of transition rules of level greater than $k$, by definition, $\mathcal{T}_{t} \left( \mathcal{C} \left( E ; t \right) \vert c(t), \mathcal E \right)$ contains all the points up to time $t$ that have non-trivial anyons caused by $\mathcal{C} \left( E ; t \right)$. This set clearly depends on fusion outcomes. We next define the trajectory domain of an error cluster, which incorporates the trajectories over all possible fusion outcomes.

% ALTERNATE DEFINITION TO GO WITH THE NEW CAUSAL REGION DEFINITION

\begin{definition}
The trajectory domain $\mathcal D_{t'} \left( \mathcal{C} \left( E; t \right) \vert  \mathcal{E} \right)$ of an error cluster $ \mathcal{C} \left( E ; t \right)$ where $t' \geq t$ is defined as the union over all fusion histories $c(t')$ and over all $ \{ E_{(i)} \} \subseteq \mathcal{C} \left( E ; t \right) $ of the trajectories $\mathcal{T}_{t'} \left( \mathcal{C} \left( E; t \right) \backslash \{ E_{(i)} \} \vert c(t'), \mathcal{E} \backslash \{ E_{(i)} \} \right)$.
\end{definition}

% ORIGINAL DEFINITION

%\begin{definition}
%The trajectory domain $\mathcal D_t \left( \mathcal{C} \left( E; t \right) \vert  \mathcal{E} \right)$ of an error cluster $ \mathcal{C} \left( E ; t \right)$ is defined as the union over all fusion histories $c(t)$ and over all $ E' \in \mathcal{C} \left( E ; t_{f} \right) \bigcup \{ \emptyset \}$ of the trajectories $\mathcal{T}_t \left( \mathcal{C} \left( E; t \right) \backslash E' \vert c(t), \mathcal{E} \backslash E' \right)$.
%\end{definition}

%We define the {\em trajectory domain} of an error $e$ simply by considering the union of all trajectories obtained from the different possible fusion outcomes $\mathcal T_t(e\vert \mathcal E) = \bigcup_{c(t)\vert P(c(t))\neq 0} \mathcal{T}_{t} \left( e\vert c(t), \mathcal E \right)$.  The definitions of  the trajectory domains $\mathcal T(e\vert \mathcal E)$, $ \mathcal{T}_{t} \left( \{ E_{\left( i \right) } \}\vert  \mathcal E \right)$, and $ \mathcal{T} \left( \{ E_{\left( i \right) } \}\vert \mathcal E \right)$ follow straightforwardly. 

In the previous section, we introduced the notion of causally-linked cluster,  a god's eye view of the error history which speaks about the errors directly. The present section introduces the notion of trajectories, which provide a physicist's perspective on the same history, described in terms of the excitations that are caused by the errors. The following lemma establishes a relation between these two points of view.

%%%%% MODIFIED FOR THE NEW DEFINITION OF CAUSAL REGION

\begin{proposition}
In the absence of transition rules of level greater than $k$, the causal region of $\mathcal{C} \left( E ; t \right)$ up to time $t$ is contained in its trajectory domain. 
\label{prop_tree}
\end{proposition}

\begin{proof}
Suppose the contrary, there exists a point $p$ with spatial coordinate at most $t$ in the causal region of $\mathcal{C} \left( E ; t \right)$ but which is not in any of the trajectory $\mathcal{T}_t \left( \mathcal{C} \left( E; t \right) \backslash \{ E_{(i)} \} \vert  c(t), \mathcal{E} \backslash \{ E_{(i)} \} \right)$. First, we notice that by definition of a causally-linked cluster, if $p$ is part of $\mathcal{C} \left( E ; t \right)$, then it cannot be part of any other distinct causally-linked cluster. We can thus consider that $\mathcal{E} = \mathcal{C} \left( E \right)$. By definition, the causal region of $\mathcal{C} \left( E ; t \right)$ is the set of points whose charge distribution is altered when errors are removed from the cluster. These charge distributions are given by  Eq. \ref{eqn_prob_s}. In the absence of transition rules of level greater than $k$, by assumption, $p$ is not part of the trajectory of any error of the form $\mathcal{C} \left( E; t \right) \backslash \{ E_{(i)} \}$. It follows that any term in the sum of equation \ref{eqn_prob_s} involved in $P \left( C \left( p \right) \neq 1, \mathcal{E} \backslash \{ E_{(i)} \} \right)$ is equal to $0$. By definition, we thus have that for any $E' \in \mathcal{C} \left( E ; t \right)$, $p \notin Z \left( E' \right)$, which shows that $p$ is not in the causal region of $\mathcal{C} \left( E ; t \right)$, a contradiction.
\end{proof}

Lastly, our analysis of the hierarchical error-correction procedure will require a coarse-grained notion of the trajectory which builds on the coarse-grained error clusters.

\begin{definition}
For any $\tau \in \mathbb{N}$, $\mathcal{T}_{\left( k \right)} \left( \tau \right) = \{ C^{k}_{\bm \rho} \vert C^{k}_{\bm \rho} \bigcap \mathcal{L}_{t} \left( \mathcal{C}^{\left( k \right)} \left( t \right) \vert  c \left( t \right), \mathcal{E} \right) \neq \emptyset  \}$ such that $t \in \left[ \right.  \tau U^{k}+\left( 1-f_{c} \right) bU^{k-1}, \left( \tau+1 \right)U^{k} \left. \right)$.
\end{definition}

Essentially, $\mathcal{T}_{\left( k \right)} \left( \tau \right)$ is the set of level-$k$ colonies containing excitations belonging to a level-$k$ or larger causally-linked clusters during the $\tau$\emph{th} level-$k$ renormalized working period. Importantly, we only consider overlaps that persist beyond time $\tau U^{k}+\left( 1-f_{c} \right) bU^{k-1}$ instead of considering the beginning of the working period of $\tau U^{k}$. Consider the case where a level-$k$ transition rule is applied at time $\tau U^k$, moving a non-trivial charge outside of $C^k_{\bm \rho}$, and such that it interacts with an actual error of lower levels on its path. It is possible that the system nevertheless evolves in a manner such that all the anyons of the smaller actual error fuse to the vacuum in a time scale too short to influence the level-$k$ syndrome of the colony. We do not wish to include $C^k_{\bm \rho}$ in $\mathcal{T}_{\left( k \right)} \left( \tau \right)$ in order to be coherent with definition \ref{def_success_rule}. The specific choice for the time $\left( 1-f_{c} \right) bU^{k-1}$ will become clear in light of property \ref{def_k_faithful}.

\subsection{Syndromes}

The total charge contained within a fixed region is in general not precisely defined before fusing all the anyons within that region together. While the intuitive idea of a level-$k$ syndrome is to report the total charge of a level-$k$ colony, care must be taken to formally determine when a syndrome reports a valid charge.

%%% NEW DEFINITION OF RENORMALIZED SYNDROME

Suppose that the system has evolved to time $t$, and is described by the state $\vert \psi (c (t), \mathcal E ) \rangle$. Consider the state $\vert \psi' (c'(t), \mathcal E' ) \rangle$ with $\mathcal E' = \mathcal{C}^{(k)} (t) $ and such that for any $t_0 \leq t$, $c'(\bm r, t_0) = c(\bm r, t_0)$ if $(\bm r, t_0)$ is in the causal region of one of the causally-linked clusters in $\mathcal C^{\left( k \right)} (t)$, and $c'(\bm r, t_0) = \mathbbm 1$ otherwise. Let $\{a_{1}, \dots, a_{n} \}$ denotes the charge of the anyons in a level-$k$ colony $C^{k}_{\bm \rho}$ at time $t$ and let $\Pi_{a}^{a_{1}, \dots, a_{n}}$ be the projector of their fusion outcome onto a total charge $a$ (by bringing them all at the center of the colony). We define the level-$k$ charge distribution by $\rho_{\bm \rho}^{k} \left( a , t \right) = \langle \psi' \left( c' \left( t \right), \mathcal{E'} \right) \vert \Pi_{a}^{a_{1}, \dots, a_{n}} \vert \psi' \left( c' \left( t \right), \mathcal{E'} \right) \rangle$. In the case where no anyon is present, then $\rho_{\bm \rho}^{k} \left( a , t \right) = \delta_{1, a}$.

In the above definition, we assumed that removing a causally-linked cluster of lower level does not modify the topological charges present in the causal region of any causally-linked cluster in $\mathcal{C}^{(k)} (t)$ for a \emph{specific} fusion history. Although our definition of causal region assures us that it cannot have an effect for the charge distribution averaged over all possible fusion outcomes, it is technically possible that for some specific fusion histories, some of the topological charges of a cluster in $\mathcal{C}^{(k)} (t)$ are modified, even though it seems quite unlikely. To avoid such difficulties, we simply consider that if this happens, then the causally-linked clusters in play are merged together. Alternatively, one could modify the definition of the causal region of an actual error to consider the charge probability density of a site for every possible charge history. We do not consider such a definition here to avoid unnecessary complications.

\begin{definition}
If at a time $t$ such that $t + \frac{1}{2} = (\tau+1) U^k$, $\rho_{\bm \rho}^{k} \left( a , t \right) = \delta_{a,c}$ for some topological charge $c$ (including the vacuum), and if the reported syndrome $s_{k,c} \left( \bm \rho, \tau \right)$ corresponds to $c$, then we say that $s_{k,c} \left( \bm \rho , \tau \right)$ is valid. The same definition of validity holds for $s_{k,n} $.
\end{definition}

\emph{A priori}, this notion of a valid syndrome may not seem very useful, as there can obviously be anyons part of clusters of lower level having support in $C^{k}_{\bm \rho}$. However, it will be shown below that for non-cyclic anyonic models, any such cluster will eventually be corrected on a time scale which is too short to influence higher level syndromes. Furthermore, the total charge of \emph{all} the anyons in $C^{k}_{\bm \rho}$ is subject to rapid fluctuations, as anyons part of lower-level clusters can quickly appear and disappear along the border of the colony.

Suppose that a level-$k$ colony $C^{k}_{\bm \rho}$ has a valid, non-trivial syndrome $s_{k,c}$ and that its $8$ neighbouring colonies, collectively denoted by $\{ C^{k}_{\langle \bm \rho \rangle} \}$, also have valid syndromes $s_{k,n}$. The level-$k$ transitions rules $M_{\bm \rho}^{ \left( \alpha, \beta \right), k}$ (using the notation of appendix \ref{AppendixA}) are applied at times $t_0$ such that $t_{0} + \frac{1}{2} = (\tau + 1) U^k$. If, based on this syndrome pattern, the transition rules dictate that the charge contained in the central colony is to be displaced to one of the neighbouring colonies, we say that the transition rule is {\em active}. Note that by definition, a transition rule is not called active when the syndrome of the central colony or one of its neighbours is invalid, or when the central colony contains the vacuum.  

\begin{definition}
The active transition rule $M_{\bm \rho}^{ \left( \alpha, \beta \right), k}$ is said to have been successfully applied if for the next level-$k$ working period, $s_{k,c} \left( \bm \rho, \tau + 1 \right)$ is valid and reports the vacuum.
\label{def_success_rule}
\end{definition}

In the above definitions, level-$0$ colonies are to be interpreted as single physical sites. A level-$0$ syndrome is thus valid if it reports the anyonic charge of the physical site under consideration. A level-$0$ active transition rule is said to have been applied successfully if after its application the physical site under consideration contains the vacuum.

\section{Threshold for Non-Cyclic Anyonic Models}
\label{sec_threshold}

In this section, we sketch a proof that the error correction procedure described in the previous section and in the appendix has a fault-tolerant threshold for non-cyclic anyonic models. For the remainder of this section, we suppose that $a \geq 3$ and $b > 2\left[4 \left( 3 \textnormal{diam} \left( G_{\mathcal{A}} \right) + 1 \right) Q + 1 \right] $, and we choose $f_{c} b = \break b - \left( 4 \left( 3 \textnormal{diam} \left( G_{\mathcal{A}} \right) + 1 \right) Q + 1 \right)$ and $f_{n} b = 4 \left( 3\textnormal{diam} \left( G_{\mathcal{A}} \right) + 1 \right) Q + 1 $ (so that $b$ is large enough to ensure that $f_c > f_n$). Note that these values are in accordance with all previous bounds given for $a$ and $b$.

We begin by defining three key properties of a correction algorithm which will be central in proving that it possesses a threshold.

\begin{definition}
\emph{$k$-locality.} If, for all times $t$ and all level-$k'$ actual error $E$ with $k' \leq k$, the trajectory domain $\mathcal D_t \left( \tilde{\mathcal{C}} \left( E ; t  \right), \mathcal{\tilde E} \left( E \right)   \right)$  is contained in a block of $2 \times 2$ level-$k'$ colonies, we say that the algorithm is $k$-local.
\label{def_k_local}
\end{definition}

\begin{definition}
\emph{$k$-faithfulness.} 
Suppose that a level $k'$ colony $C^{k'}_{\bm \rho}$, with $k'\leq k$, and its 8 neighbours  $\{ C^{k'}_{\langle \bm \rho \rangle} \}$ do not overlap with errors of level $k'$ or greater for a full level-$k'$ working period  $\left[ \right. \tau U^{k'}, \left( \tau+1 \right)U^{k'} \left. \right)$. If there exists $t_{c} \leq  \left( 1 - f_{c} \right) b U^{k'-1} + \tau U^{k'}$ such that for any $t \geq t_{c}$ in that working period, any point $p= (\bm r, t) $ in the trajectory domain on that colony $\break \mathcal{D}_{t} \left( \mathcal{C}^{\left( k \right)} \left( t_{c} \right), \mathcal{E} \right) \bigcap C^{k'}_{\bm \rho}$ is in the physical site at the centre of $C^{k'}_{\bm \rho}$,  we say that the algorithm is $k$-faithful.
\label{def_k_faithful}
\end{definition}

\begin{definition}
\emph{$k$-successfulness.} If for any $k' \leq k$, in the absence of level-$k'$ or larger actual error and applied transition rules of level higher than $k'$ on $C^{k'}_{\bm \rho} \bigcup \{ C^{k'}_{\langle \bm \rho \rangle} \}$, any level-$k'$ active transition rule $M^{\left( \alpha, \beta \right),k'}_{\bm \rho}$ which is applied $\textnormal{diam} \left( G_{\mathcal{A}} \right)$ successive times is successful, we say that our algorithm is $k$-successful.
\label{def_k_successful}
\end{definition}

As we will show below,  these three properties ensure that the level-$k$ transition rules react to level-$k$ causally-linked clusters essentially the same way as  level-$0$ rules react to level-$0$ actual error at the physical level. Loosely speaking, $k$-locality asserts that local errors are kept local by the updates rules. $k$-faithfulness asserts that in the absence of high-level errors, the transition rules will bring all excitations in the centre of the colony rapidly enough to trigger a valid non-trivial syndrome. Note that the excitations caused by low-levels clusters are excluded from this demand since only the coarse grained error cluster is considered, {\it c.f.}, Def.~\ref{def:coarse_grained_cluster}. Lastly, $k$-successfulness essentially demands that the transition rules correct errors, but allows repeated attempts to take into account the possibility that a non-abelian anyon and its anti-particle do not fuse to the vacuum. These consequences are stated formally in the following lemma.

\begin{lemma}
Assuming the correction algorithm is $k$-local, $k$-faithful, and $k$-successful, then for $k' \leq k$, it also has the following properties:
\begin{enumerate}
\item  If $C^{k'}_{\bm \rho} \bigcup \{ C^{k'}_{\langle \bm \rho \rangle} \}$ is free of level-$k'$ or larger actual errors,  that $ C^{k'}_{\bm \rho} $ is not affected by transition rules of level $k'$ or greater, and that $\{ C^{k'}_{\langle \bm \rho \rangle} \}$ are not affected by transition rules of level greater than $k'$ for a full level-$k'$ working period  labelled $\tau$, then $C^{k'}_{\bm \rho} \in \mathcal{T}_{\left( k' \right)} \left( \tau \right) \iff C^{k'}_{\bm \rho} \in \mathcal{T}_{\left( k' \right)} \left( \tau + 1 \right)$.
\label{conjecture_k_trajectory_evolution}
\item If  $C^{k'}_{\bm \rho} \bigcup \{ C^{k'}_{\langle \bm \rho \rangle} \}$ is free of level-$k'$ or larger actual errors for a full level-$k'$ working period  labelled $\tau$, then $s_{k',c} \left( \bm \rho, \tau \right)$ and $s_{k',n} \left( \bm \rho,  \tau \right)$ are valid.  
\label{lemma_valid_syndrome}
\item If $E$ is a level-$k'$ actual error and ${\rm lvl}\left(\mathcal{C} \left( E ; t \right)\right) \leq k'$ for all $t$, then the causal region of $\mathcal{C} \left( E \right)$ is contained within a space-time box of size $2Q^{k'} \times 2Q^{k'} \times \left( 2  \textnormal{diam} \left( G_{\mathcal{A}} \right) + 3 \right) U^{k'}$.
\label{lemma_causal_region}
\end{enumerate}
\end{lemma}

Note that with our choice of $b > 2\left[4 \left( 3 \textnormal{diam} \left( G_{\mathcal{A}} \right) + 1 \right) Q + 1 \right] $, the box containing a level-$k'$ causally-linked cluster in the third consequence is sufficiently small to prevent interaction with another level-$k'$ actual error.

\begin{proof}

We prove these properties recursively: we first show that they each hold individually for $k'=0$, and then assume that they all hold for $k'$ to prove that they each hold for $k'+1$, up to $k' + 1 = k $. 

%The first two properties are proven in a similar way. For $k'=0$, these properties follow directly from the definition of the transition rules and the absence of errors. Suppose now that the properties hold true for $k'-1$. In the absence of errors, all the charges concentrate in the level $k'$ colony centres after at most $(Q+1)U^{k'-1}$ time steps. Indeed, a level-$(k'-1)$ transition rule is applied every $U^{k'-1}$ time steps and there are at most $Q+1$ level-$(k'-1)$ colonies to cross in order to reach the center. The presence of low-level errors could possibly deviate the charge towards the center of another colony or, once it has reached the colony center, it could move the charge away from the colony center.

%The possibility of deviating the charge towards another colony can only occur at the beginning of the working period, when the charge is near the  colony boundary. Indeed, $k$-faithfullness guarantees that the charge will reach the center at most at time $\left( 1 - f_{c} \right) b U^{k'-1}$. But by definition, the coarse-grained trajectory is insensitive to the position of the charges up to that time, and so it is not sensitive to low-level errors. 

\begin{enumerate}

%%% Statement 1

\item Suppose that $C^{0}_{\bm r}$ is neither affected by any actual errors at time $t+1$ nor by any level-$0$ or higher transition rules at time $t + \frac{1}{2}$. The charge state of $C^{0}_{\bm r}$ at time $t+1$ is the same as that at time $t$, since there exists no other process that can change the topological charge contained in $C^{0}_{\bm r}$ between time $t$ and $t+1$. We thus have that $C^{0}_{\bm r} \in \mathcal{T}_{\left( 0 \right)} \left( t  \right) \iff C^{0}_{\bm r} \in \mathcal{T}_{\left( 0 \right)} \left( t + 1  \right)$.

%Assume now that the statement if true for $k' < k $. We show that it is also true for $k' + 1$. 
%Consider the level-$\left( k'+1 \right)$ colony $C^{k' + 1}_{\bm \rho}$ and the two consecutive level-$\left( k' + 1 \right)$ working periods starting at times $tU^{k'+1}$ and $\left( t + 1 \right)U^{k'+1}$ during which no level-$\left( k'+1 \right)$ or larger actual errors and no level-$\left( k' + 2 \right)$ or larger transition rules affects any of the level-$\left( k' + 1 \right)$ colonies in $\{ C^{k' + 1}_{\bm \rho} \bigcup \{ C^{k'+1}_{\langle \bm \rho \rangle} \} \}$. Suppose furthermore that no transition rule of the from $M^{\left( \alpha, \beta \right),k'+1}_{x,y}$ or $M^{\left( \alpha, \beta \right),k'+1}_{x - \alpha, y - \beta}$ is applied at time $\left( t + 1 \right)U^{k'+1} - \frac{1}{2}$.

Consider first the case $C^{k' + 1}_{\bm \rho} \in \mathcal{T}_{\left( k'+1 \right)} \left( \tau \right)$. Using $\left( k' + 1 \right)$-faithfulness, we first notice that there exists $t_{c} \leq \left(1-f_{c} \right)bU^{k'} + \tau U^{k'+1}$ such that for any $t' \geq t_{c}$, any point in $\displaystyle  \mathcal{L}_{t'} \left( \mathcal{C}^{\left( k+1 \right)} \left( t_{c} \right)  \vert  c \left( t' \right), \mathcal{E} \right) \bigcap C^{k' + 1}_{\bm \rho} $ is in the physical site at the centre of $C^{k' + 1}_{\bm \rho}$. Since $C^{k' + 1}_{\bm \rho} \in \mathcal{T}_{\left( k'+1 \right)} \left( \tau \right)$, we can always choose $t_{c}$ such that $c \left( {\bf r}_0, t_{c} \right) \neq 1$, where ${\bf r}_0$ labels the centre of $C^{k' + 1}_{\bm \rho}$. We now need to consider how errors of level $k'$ or less could move the charge away from the colony centre after time $t_c$.

Consider the level-$k'$ actual error $F$ containing its first error happening at time $t_{f}$ with $t_{c} \leq t_{f} < \left( \tau + 1 \right) U^{k'+1}$, such that for any $t'$ and any charge measurements $c \left( t' \right)$, the trajectory \newline $\mathcal{T}_{t'} \left( \mathcal{C} \left( F ; t' \right) \vert  c \left( t' \right), \mathcal{E} \right)  \bigcap \left( C^{k'}_{\bm \rho} \bigcup \{ C^{k'}_{\langle \bm \rho \rangle} \} \right) = \emptyset$, with $C^{k'}_{\bm \rho }$ the level-$k'$ colony at the centre of $C^{k' + 1}_{\bm \rho}$. Using $k'$-locality, we know that the trajectory domain of this error will never overlap with $C^{k'}_{\bm \rho}$. Moreover, using the inductive hypothesis on the third consequence, we know that this error will be corrected before it can interact with another level-$k'$ error.

Suppose, on the other hand, that there exists $t'$ and $c \left( t' \right)$  such that $\break \mathcal{T}_{t'} \left( \mathcal{C} \left( F ; t' \right) \vert c \left( t' \right), \mathcal{E} \right) \bigcap \left( C^{k'}_{\bm \rho} \bigcup \{ C^{k'}_{\langle \bm \rho \rangle} \} \right) \neq \emptyset$. We first note that $F$ is $\break \left( aQ^{k'}, aQ^{k'}, bU^{k'} \right)$-separated from any other level-$k'$ actual error. We thus have that for the $\left( b - 1 \right)$ following level-$k'$ working periods after the last error of $F$, we need not consider other level-$k'$ actual errors. Using $k'$-locality (taking into account the central charge), we find that for any $t' \in \left[ t_f, t_f + bU^{k'} \right]$, $\mathcal{L}_{t'} \left( \mathcal{C} \left( F ; t' \right) \vert c \left( t' \right), \mathcal{E} \right) \bigcap C^{k' + 1}_{\bm \rho}$ is contained within a block $\mathcal{B}'$ of $2 \times 2$ level-$k'$ colonies which we suppose contains $C^{k'}_{\bm \rho}$ (the other case is handled similarly as in the previous paragraph). Using the inductive hypothesis on the second consequence, we find that at least for the next $\left( b - 1 \right)$ level-$k'$ working periods after the last error of $F$, the level-$k'$ syndromes are valid in $\mathcal{B}'$. Because $\mathcal B'$ is $2\times 2$ and using $k'$-faithfulness, it will take at most $ 2 \text{diam} \left( G_{\mathcal{A}} \right) + 3 $ level-$k'$ working periods before the set of locations $\mathcal{L}_{t'} \left( \mathcal{C} \left( F ; t' \right)  \vert c \left( t' \right), \mathcal{E}  \right) \bigcap C^{k' + 1}_{\bm \rho}$ is null  or contained in the physical site at the centre of $C^{k' + 1}_{\bm \rho}$. This is because after the level-$k'$ syndromes in $\mathcal B'$ are valid, which can take up to 2 level-$k'$ working period, the application of 2 rounds of successful active level-$k'$ transition rules suffices to bring the trajectory inside $C^{k'}_{\bm \rho}$ (see appendix \ref{AppendixB}). An additional working period ensures that all the non trivial anyons  reach the physical site at the center of $C^{k' + 1}_{\bm \rho}$, using $k'$-faithfulness. The same analysis works for lower-level errors as well.

The same reasoning can also be used for the case where $C^{k' + 1}_{\bm \rho} \notin \mathcal{T}_{\left( k'+1 \right)} \left( t \right)$. In that case, we find that the central site of $C^{k' + 1}_{\bm \rho}$ contains the vacuum at some time $t_c$ and will return to the vacuum at most $ 2 \text{diam} \left( G_{\mathcal{A}} \right) + 3 $ time steps after the appearance of a level-$k'$ or less error. Moreover, the previous argument shows that errors of level $k'$ or less cannot transport the charge from the centre of a neighbouring  $(k'+1)$-colony into  $C^{k' + 1}_{\bm \rho}$, so $C^{k' + 1}_{\bm \rho} \notin \mathcal{T}^{\left( k'+1 \right)} \left( t + 1 \right)$.

%%% Statement 2

\item The property is true for $k' = 0$, since if no error affects a site, the reported level-$0$ syndrome is valid. 

Suppose next that the property is true for $k' < k$, and consider $C^{k' + 1}_{\bm \rho}$, a level-$( k'+1 )$ colony. Using $(k' + 1)$-faithfulness, there exists $t_c \leq \tau U^{k'} + (1 - f_c )bU^{k'}$ such that for any charge measurement $c \left( t_c \right)$, any point in $\displaystyle \break \mathcal{L}_{t_c} \left( \mathcal{C}^{\left( k+1 \right)} \left( t_{c} \right)  \vert  c \left( t_c \right), \mathcal{E} \right) \bigcap C^{k' + 1}_{\bm \rho} $ is in the physical site at the centre of $C^{k' + 1}_{\bm \rho}$. After time $t_c$ and up to time $(\tau+1)U^{k'+1}$, if no level-$k'$ actual error affects the level-$k'$ colony at the centre of the level-$(k'+1)$ colony $C^{k' + 1}_{\bm \rho}$, the level-$k'$ syndromes reported at the centre of the level-$(k'+1)$ colony are valid, by using the inductive hypothesis on property 2. Furthermore, by using the inductive hypothesis of property 1 for $k'$, level-$(k' -1 )$ and lower actual errors cannot bring the $k'$-trajectory outside of the level-$k'$ colony at the centre of $C^{k' + 1}_{\bm \rho}$. Given the procedure to report level-$(k'+1)$ syndromes used by our algorithm, the level-$(k'+1)$ syndromes reported by $C^{k' + 1}_{\bm \rho}$ at the end of the level-$(k'+1)$ working period $t$ is thus valid. 

Suppose, on the other hand, that a level-$k'$ actual error $F$ affects the level-$k'$ colony at the centre of $C^{k' + 1}_{\bm \rho}$ with its first error happening at time $t_f$ such that $\tau U^{k'+1}+ t_{c} \leq t_f  <(\tau+1) U^{k'+1}$. First note that by definition, each of the $b$ time bin of duration $bU^{k'}$ of the level-$(k'+1)$ working period can contain at most one of these errors. While this error can move the charge away from the centre of the colony, using the same argument as above, we know that after at most $ 2 \text{diam} \left( G_{\mathcal{A}} \right) + 3 $ level-$k'$ working periods later, the set of locations $\mathcal{L}_{t'} \left( \mathcal{C} \left( F ; t' \right)  \vert c \left( t' \right), \mathcal{E}  \right) \bigcap C^{k' + 1}_{\bm \rho}$ is null  or contained in the physical site at the centre of $C^{k' + 1}_{\bm \rho}$. Given our choice of $f_{c} b =  b - \left( 4 \left( 3 \textnormal{diam} \left( G_{\mathcal{A}} \right) + 1 \right) Q + 1 \right)$ and $f_{n} b = 4 \left( 3\textnormal{diam} \left( G_{\mathcal{A}} \right) + 1 \right) Q + 1 $, this leaves enough time in each of the remaining time bins (those after $t_c$) to report a valid syndrome, and hence the level-$(k'+1)$ syndrome is valid.

%%% Statement 3

\item We first ignore transition rules of level higher than $\left( k' + 1 \right)$. Suppose that $F$ is the only level-$(k' + 1)$ actual error of $\mathcal{E}$, and let $t_{0}$ be the time of the first error in $F$. In that case, $\tilde{\mathcal C} \left( F; t' \right) = \mathcal{C} \left( F ; t' \right)$. Using $(k'+1)$-locality and proposition \ref{prop_tree}, we find that the causal region of $\mathcal{C} \left( F, t' \right)$ is contained within a block $\mathcal{B}$ of $2 \times 2$ level-$(k'+1)$ colonies.

The previous property ensures that for the $b -1$ complete level-$(k'+1)$ working periods after $F$, level-$(k'+1)$ syndromes in $\mathcal{B}$ and neighbouring level-$(k'+1)$ colonies are valid. Using  property 1 for $k'+1$, the $(k'+1)$-trajectory of $\mathcal{C} \left( F, t' \right)$ can only be modified by the application of level-$(k'+1)$ transition rules. Using $(k'+1)$-successfulness, active level-$(k'+1)$ transition rules in $\mathcal{B}$ are successfully applied with at most $\textnormal{diam} \left( G_{\mathcal{A}} \right)$ tries. By inspection of the transition rules, for any set of charge measurements $c$, after applying at most two rounds of successful level-$(k'+1)$ transition rules, $\mathcal{T}_{t'} \left( \mathcal{C} \left( F ; t' \right)\vert c \left( t' \right) , \mathcal{E}  \right)$ is contained within a single level-$(k'+1)$ colony $C^{k' + 1}_{\bm \rho}$, where $t' \geq t_{0} + \left( 2 \textnormal{diam} \left( G_{\mathcal{A}} \right) + 2 \right) U^{k}$ (see appendix \ref{AppendixB}). Using $(k'+1)$-faithfulness, we find that there exists $t_c \leq t_{0} + \left( 2 \textnormal{diam} \left( G_{\mathcal{A}} \right) + 3 \right) U^{k'+1}$ such that $\mathcal{T}_{t_c} \left( \mathcal{C} \left( F; t_c \right) \vert c \left( t_c \right) , \mathcal{E}  \right)$ is contained within a single physical site at the centre of $C^{k' + 1}_{\bm \rho}$. Provided the linear size of the system is larger than $2Q^{k'+1}$ (so the worldline of the charges is homologically trivial), the resulting charge is the vacuum by charge conservation. The same reasoning can be repeated using the set of actual errors $\mathcal{E} \backslash \{ E_{(i)} \}$, for any subset of actual errors $\{ E_{(i)} \} \subseteq \mathcal{C} \left( F \right)$. Using proposition \ref{prop_tree}, we conclude that the causal region of $\mathcal{C} \left( F \right)$ is contained within a space-time box of size $2Q^{k'+1} \times 2Q^{k'+1} \times \left( 2  \textnormal{diam} \left( G_{\mathcal{A}} \right) + 3 \right) U^{k'+1}$.

If there are more than one such level-$(k'+1)$ actual errors affecting the system, the previous reasoning can be performed for each of their respective causally-linked clusters. Consider $F$, the first such level-$(k'+1)$ actual error to happen (pick any one if more than one such error happen simultaneously). Since level-$(k'+1)$ actual errors are $\left( aQ^{k'+1}, aQ^{k'+1}, bU^{k'+1} \right)$-separated and with our choice of $b > 2\left[4 \left( 3 \textnormal{diam} \left( G_{\mathcal{A}} \right) + 1 \right) Q + 1\right] $, the above argument shows that no actual error in the causally-linked cluster of $F$ can have interacted with any other level-$(k'+1)$ actual error.

We now consider transition rules of level higher than $(k'+1)$. We show that level-$(k'+1)$ causally-linked clusters cannot cause the application of transition rules of level larger than $(k'+2)$. Consider a level-$(k'+2)$ colony $C^{k'+2}_{\bm \rho_{k'+2}}$ and working period $\tau_{k'+2}$. Suppose that no actual error of level larger than $(k'+1)$ happens. Using the validity of property 2 for $(k'+1)$, if no level-$(k'+1)$ actual error affects the level-$(k'+1)$ colony $C^{k'+1}_{\bm \rho_{k'+1}}$ at the centre of $C^{k'+2}_{\bm \rho_{k'+2}}$ during a level-$(k'+1)$ working period, the corresponding level-$(k'+1)$ syndrome reports the vacuum. Suppose, on the other hand, that a level-$(k'+1)$ actual error affects  $C^{k'+1}_{\bm \rho_{k'+1}}$ during any bin of $b$ level-$(k'+1)$ working periods of $\tau_{k'+2}$, except the last one. Since a level-$(k'+2)$ transition rule can only be applied after the last bin of $b$ level-$(k'+1)$ working periods, no level-$(k'+2)$ or higher transition rule can be applied during the considered time bin. We can thus use the previous reasoning and find that the causal region of such a level-$(k'+1)$ causally-linked cluster is contained within a box of size $2Q^{k'+1} \times 2Q^{k'+1} \times \left( 2  \textnormal{diam} \left( G_{\mathcal{A}} \right) + 3 \right) U^{k'+1}$. Given the values of $f_{c} = \frac{ b - \left( 4 \left( 3 \textnormal{diam} \left( G_{\mathcal{A}} \right) + 1 \right) Q + 1 \right) }{b}$ and $f_{n} = \frac{4 \left( 3\textnormal{diam} \left( G_{\mathcal{A}} \right) + 1 \right) Q + 1}{b} $, we thus find that all such bins of $b$ level-$(k'+1)$ working periods report the vacuum. Since the $b-1$ first level-$k'+1$ working periods report the vacuum, the reported level-$(k'+2)$ syndromes for $C^{k'+2}_{\bm \rho_{k'+2}}$ for working period $\tau_{k'+2}$ are the vacuum. Noting that no level-$(k'+2)$ transition rule is applied if the vacuum charge is reported in a level-$(k'+2)$ colony completes the proof for level-$(k'+2)$ transition rules. Since every level-$(k'+2)$ working period reports the vacuum charge, the same conclusion is reached for higher level transition rules as well.

Lastly, the transition rules of level larger than $k'+1$ cannot affect level-$(k'+1)$ causally-linked clusters. Suppose that a level-$(k'+1)$ causally-linked cluster is affected by the application of a level-$(k'+2)$ or larger transition rule. Then previous reasoning shows that transition rule is caused by a level-$(k'+2)$ or larger causally-linked cluster. By definition of a causally-linked cluster, the actual errors being affected by the transition rule and the actual errors having caused the application of the transition rule are part of the same causally-linked cluster. Since this causally-linked cluster is of level-$(k'+2)$ or larger, we find a contradiction.

\end{enumerate}
\end{proof}

\begin{lemma}
The correction algorithm we present is $0$-local, $0$-faithful and $0$-successful.
\label{lemma_0_level}
\end{lemma}

\begin{proof}
Let $E$ be a level-$0$ actual error and let $\tilde{\mathcal{C}} \left( E ; t \right)$ be its partial causally-linked cluster at time $t$. We first note that by definition, $\tilde{\mathcal{C}} \left( E; t \right) = E$, and that $\mathcal{\tilde E} \left( E \right) = E$. It is straightforward to verify that for any time $t$ and any charge measurement $c \left( t \right)$, we have that $\mathcal{T}_{t} \left( E\vert c \left( t \right), E\right)$ is contained within a unique block of $2 \times 2$ sites (see appendix \ref{AppendixB}). We thus have that our algorithm is $0$-local.

Since a level-$0$ colony $C^{0}_{\bm r}$ consists of a single site, it is obvious that our algorithm is $0$-faithful.

Inspection of level-$0$ transition rules reveals that an active level-$0$ transition rule is always applied successfully at its first try, showing $0$-successfulness.
\end{proof}

\begin{lemma}
Suppose that the algorithm is $k$-local, $k$-faithful and $k$-successful. Then it is also $\left( k + 1 \right)$-faithful.
\label{lemma_faithful}
\end{lemma}

\begin{proof}

In order to follow $\mathcal{T}_{t'} \left( \mathcal{C}^{\left( k+1 \right)} \left( t' \right)\vert c \left( t' \right), \mathcal{E}  \right)$ inside $C^{k+1}_{\bm \rho}$ as a function of time, we begin by considering the temporal evolution of the $k$-trajectory $\mathcal{T}_{\left( k \right)} \left( \tau' \right)$ restricted to the level-$k$ colonies inside $C^{k+1}_{\bm \rho}$ and neighbouring level-$(k+1)$ colonies, with $t' \in [ \tau' U^k + (1 - f_c )bU^{k-1}, (\tau' + 1)U^k )$. Using property 1, we only need to consider the application of level-$k$ transition rules and the effects of level-$k$ actual errors. We wish to find a bound on the number of applications of rounds of level-$k$ transition rules before $\mathcal{T}_{\left( k \right)}$ is contained within the level-$k$ colony at the centre of $C^{k+1}_{\bm \rho}$.

Consider first the case where no level-$k$ actual errors happen during the working period in $C^{k+1}_{\bm \rho}$ nor any of the neighbouring colonies.  After one round of successful application of the level-$k$ transition rules, all the level-$k$ colonies on the border of $C^{k+1}_{\bm \rho}$ contain no anyons from causally-linked cluster of level-$k$ or higher. Using $k$-successfulness, this takes at most $\textnormal{diam} (G_{\mathcal A})$ level-$k$ working periods. Past this stage, one only needs to take into consideration the anyons inside $C^{k+1}_{\bm \rho}$. Careful consideration of the transition rules reveals that the anyons in the farthest level-$k$ colony from the centre of $C^{k+1}_{\bm \rho}$ in Manhattan distance will be moved towards the centre by a single level-$k$ colony every round of successful application of level-$k$ transition rule. Since the farthest level-$k$ colony inside $C^{k+1}_{\bm \rho}$ is at a distance of at most $Q$ from the centre, after the application of $Q$ successful level-$k$ transition rules, $\mathcal{T}_{\left( k \right)}$ restricted to $C^{k+1}_{\bm \rho}$ will contain only the level-$k$ colony at the centre. Using $k$-successfulness, we find that this requires no more than $(Q+1) \textnormal{diam} G_{\mathcal A}$ level-$k$ working periods.

%\begin{figure}[!h]
%\centering
%\includegraphics[width=.4 \textwidth]{rules_colony}
%\caption{ The possible transition rules, depending on the position of a site in the colony. The dots represent sites, and the arrows depict all the possible transtion rules to be applied. }
%\label{fig_rules_colony}
%\end{figure}

Now consider the case where level-$k$ actual errors do happen. Suppose first that an anyon is contained within a level-$k$ colony inside a corridor of $C^{k+1}_{\bm \rho}$ and away from the colony's border.  While a level-$k$ error can displace the anyon outside the corridor, $k$-locality and the transition rules ensure that it will remain within a box of $2 \times 2$ level-$k$ colonies containing the original level-$k$ colony. By inspection, we see that one successful application of level-$k$ transition rules will bring the $k$-trajectory back towards a corridor (which may be different than the one into which it was initially contained). Given that level-$k$ actual errors are $(aQ^k, aQ^k, bU^k)$-separated from each other, the anyon can drift towards the border of $C^{k+1}_{\bm \rho}$ at a rate of at most $1$ level-$k$ colonies every time it is displaced by $2$ level-$k$ colonies towards the corridor perpendicular to the border in question.  It is clear that the anyon cannot escape $C^{k+1}_{\bm \rho}$, since it will always get back to a corridor before reaching the border of the level-$(k+1)$ colony. Thus, we see that at worst the anyon spyrals towards the colony centre on $\pi/8$ angle trajectories. Suppose that the farthest anyon from the centre of $C^{k+1}_{\bm \rho}$ has reached a corridor and is at a distance of at most $\lfloor \frac{Q}{2} \rfloor$ from the centre. We seek bonds on the number of applications of successful level-$k$ transition rules before it gets to the level-$k$ colony at the centre of $C^{k+1}_{\bm \rho}$, valid for any set of level-$k$ actual errors happening on its trajectory. The longest path possible for the anyon is to jump onto another corridor, at a distance of at most $\frac{\lfloor \frac{Q}{2} \rfloor}{2}$ from the colony centre. This is the farthest distance the anyon can get from the colony centre in another corridor, given the separation between the level-$k$ actual errors. This requires displacing the anyon at a distance of at most $\lfloor \frac{Q}{2} \rfloor + \frac{\lfloor \frac{Q}{2} \rfloor}{2}$ level-$k$ colonies towards the corridor, considering the fact that a level-$k$ actual error can displace the anyon by one level-$k$ colony in the opposite direction of the corridor. This process can be repeated in order to reach another corridor, at a distance of at most $\frac{\lfloor \frac{Q}{2} \rfloor}{4}$ from the colony centre, in the same fashion. This process is then repeated until the anyon reaches the colony centre. The required number of applications of successful transition rules successfully displacing the anyon is then bounded above by $Q \sum_{i = 0}^{\infty} \frac{1}{2^{i}} = 2Q$.

Consider next an anyon whose $k$-trajectory is such that it will reach a level-$k$ colony in one of the corridors inside $C^{k+1}_{\bm \rho}$ (excluding the border). A similar analysis shows that such an anyon needs to be displaced by at most $2Q$ level-$k$ colonies before reaching a corridor of $C^{k+1}_{\bm \rho}$.

We next find a bound on the number of level-$k$ working periods it takes to displace the farthest anyon in a level-$(k+1)$ colony by a single level-$k$ colony. A level-$k$ actual error can happen in a block of $2 \times 2$ level-$k$ colonies neighbouring the one under consideration. It is possible that the error affects the syndromes reported by the neighbouring colonies under consideration and report non-trivial syndromes, so that no level-$k$ transition rule will be applied during the first working period. Furthermore, the level-$k$ actual error can put non-trivial charges in all the $4$ level-$k$ colonies it affects. Careful consideration of the transition rules shows that it is possible in that case that the non-trivial charges deposited by the level-$k$ actual error prevent the application of the transition rule for the next 2 rounds of application of successful level-$k$ transition rules. Using $k$-successfulness and the fact that a non-trivial syndrome $s_{k,n}$ can be reported for at most a single level-$k$ working period affected by level-$k$ actual error, we find that it can take up to $(3 \textnormal{diam} (G_{\mathcal A}) + 1)U^{k}$ time steps before the anyon under consideration is properly displaced by a level-$k$ transition rule.

Putting these results together, we find that the maximum time $t_c$ it takes after the start of the level-$(k+1)$ working period so that the $k$-trajectory restricted to $C^{k+1}_{\bm \rho}$ reaches the level-$k$ colony at its centre is bounded above by $4Q(3 \textnormal{diam} (G_{\mathcal A}) + 1))U^{k}$.

Finally, we can now use the fact that the algorithm is $k$-faithful to conclude that there exists $t^{\ast}_{c} \leq t_{c} + \left( 1- f_{c} \right) b U^{k-1}$ such that for any $t' \geq t^{\ast}_{c}$, the set of locations $\mathcal{L}_{t'} \left( \mathcal{C}^{\left( k+1 \right)} \left( t^{\ast}_{c} \right)\vert c \left( t' \right), \mathcal{E}  \right) \bigcap C^{k+1}_{\bm \rho}$ is included in the \emph{physical site} at the centre of $C^{k+1}_{\bm \rho}$. Since $\left( 1- f_{c} \right) b U^{k-1} < U^{k}$, and given the value of $f_{c} b =  b - \left( 4 \left( 3 \textnormal{diam} \left( G_{\mathcal{A}} \right) + 1 \right) Q + 1 \right)$, we conclude that the algorithm is $\left( k + 1 \right)$-faithful. This is because the previous reasoning is valid for an arbitrary set of charge measurements and actual error pattern $\mathcal E$, as long as it satisfies the conditions for the definition of $k$-fidelity to be applicable.

\end{proof}

\begin{lemma}
Suppose that the algorithm is $k$-local, $k$-faithful and $k$-successful, and that $Q > 14 \left( a + 2 \right) +7$. Then, there exist $b_{0}$ such that if $b > b_{0}$, the algorithm is $\left( k+1 \right)$-local.
\label{lemma_locality}
\end{lemma}

\begin{proof}

Let $E$ be a level-$\left( k + 1 \right)$ actual error. Since we are only interested in $\tilde{\mathcal{C}} \left( E ; t \right)$, we can suppose that $\mathcal{E} = \tilde{\mathcal{C}} \left( E \right)$. We follow $\mathcal{T}_{\left( k  \right)}$ as a function of (renormalized) time, given a specific set of measured charges $c \left( t' \right)$. By definition (see Sec.~\ref{sec_lvl_n_noise}),  $E$  fits in a set of $7 \left( a + 2 \right) + 1 \times 7 \left( a + 2 \right) + 1$ level-$k$ colonies. Using our assumption that $Q > 14 \left( a + 2 \right) + 7$, $E$ is contained in the bulk of a block $\mathcal{B}'$ of $2 \times 2$ level-$\left( k+1 \right)$ colonies, it does not overlap with the outer layer of level-$k$ colonies of $\mathcal{B}'$.

We begin by ignoring level-$(k+1)$ transition rules and level-$k$ actual errors. Using property \ref{conjecture_k_trajectory_evolution}, $\mathcal{T}_{\left( k  \right)}$ depends only on the level-$k$ transition rules. Since $\mathcal{T}_{\left( k  \right)}$ does not overlap with the outer layer of $\mathcal B'$, inspection of the level-$k$ transition rules show that $\mathcal{T}_{\left( k  \right)}$ remains contained within the bulk of $\mathcal B'$.

We now consider the effects of level-$k$ actual errors. As was noted in the proof lemma \ref{lemma_faithful}, level-$k$ actual errors can have the effect of displacing an anyon by a distance of at most 1 level-$k$ colonies towards the exterior of the level-$(k+1)$ colony every time it is displaced by 2 level-$k$ colonies towards a corridor. By choosing $\frac {\lceil \frac{Q}{2} \rceil}{2} < \lceil \frac{Q-7\left( a + 2 \right)-3}{2} \rceil$, basic geometrical considerations (see figure \ref{fig_colony_contained}) shows that $\mathcal{T}_{\left( k \right)} $ stays contained within the bulk of $\mathcal B'$. This last constraint ensures that $\mathcal{T}_{\left( k \right)}$ reaches a corridor in the bulk of a level-$\left( k + 1 \right)$ colony of $\mathcal{B}'$ before it can escape $\mathcal{B}'$. As was argued in the proof of the previous lemma, once an anyon has reached the corridor inside a level-$(k+1)$ colony, it cannot escape it if no level-$(k+1)$ transition rule is applied. A bit of algebra reveals that $Q > 14 \left( a + 2 \right) + 7 $ always satisfies the previous inequality.

\begin{figure}[!h]
\centering
\includegraphics[width=.65 \textwidth]{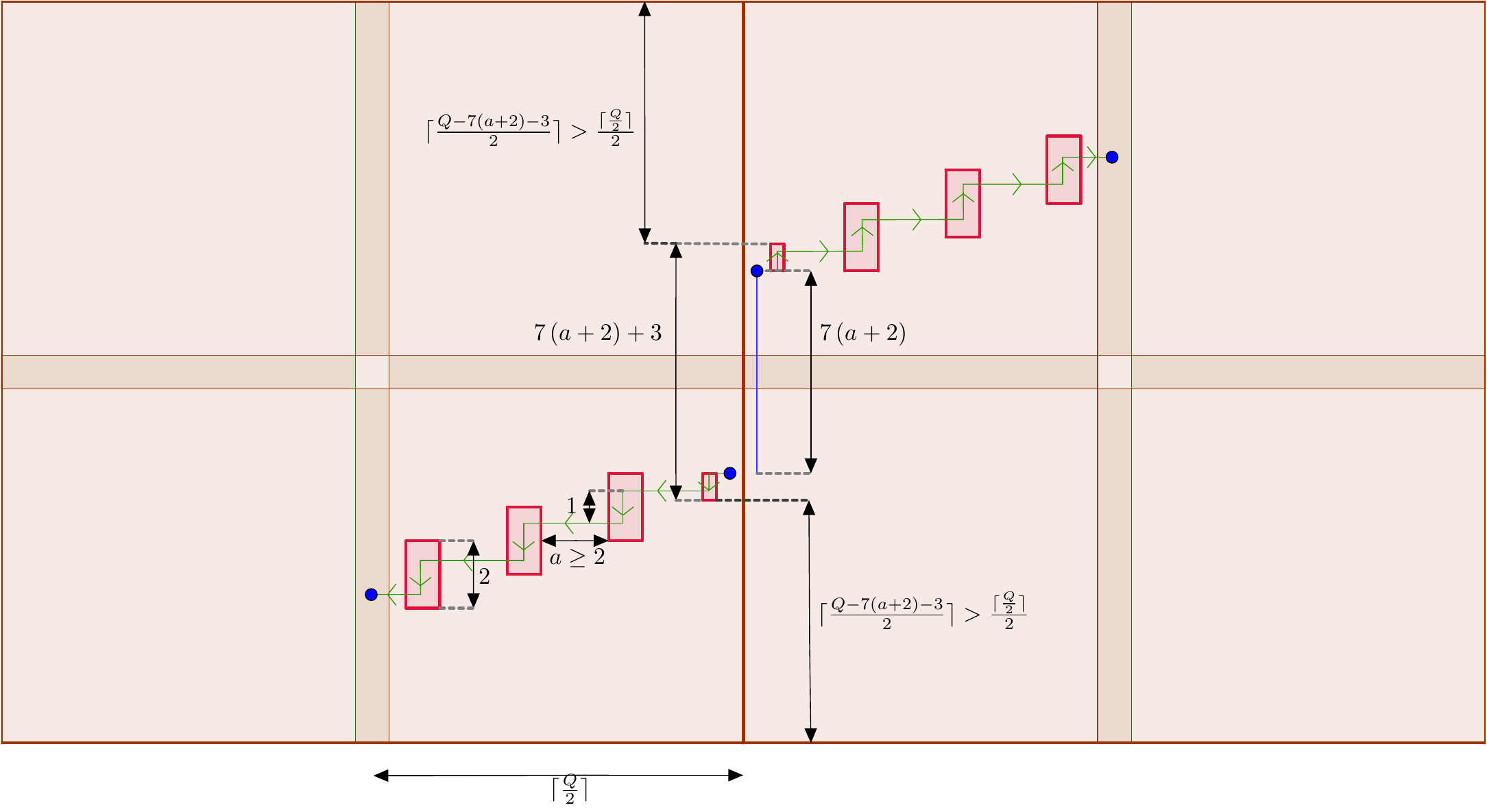}
\caption{A level-$\left(k + 1 \right)$ actual error ($E_{(k+1)}$), whose spatial extent is represented by a blue segment, fits in a space-time box of size $7 \left( a + 2 \right) Q^{k} \times 7 \left( a + 2 \right) Q^{k} \times 7 \left( b + 2 \right) U^{k}$. It leads to the creation of anyons represented by blue dots. In the absence of lower-level errors, each anyon would move towards the central corridors in a straight line. In the presence of actual level-$\left( k - 1 \right)$ and lower errors, represented by red rectangles, the anyon trajectory $\mathcal{T}_{k}$, represented by green arrows, can  expand to cover a set of $[7 \left( a+ 2 \right) + 3] \times [7 \left( a + 2 \right) + 3] $ level-$\left( k \right)$ colonies. Level-$k$ actual errors are $\left( aQ^{k}, aQ^{k}, bU^{k} \right)$-separated and can affect the spatial extension of a box containing the anyons part of $E_{\left( k +1 \right)}$. However, choosing $Q > 14 \left( a + 2 \right) + 7$ guarantees that any anyon part of the causally-linked cluster of $E_{\left( k+1 \right)}$ stay contained within a block of $2 \times 2$ level-$k+1$ colonies.}
\label{fig_colony_contained}
\end{figure}

%We begin by considering the effect of level-$\left( k - 1 \right)$ and lower actual errors on $\mathcal{T}_{\left( k  \right)}$. By using lemma \ref{lemma_k_trajectory_evolution}, we find that under the effect of level-$\left( k -1 \right)$ noise, $\mathcal{T}_{\left( k  \right)} \left( E_{\left( k \right)} \right)$ can expand at most to a region of $7 \left(a + 2 \right) + 3 \times 7 \left( a + 2 \right) + 3$ level-$\left( k \right)$ colonies containing $E_{\left( k \right)}$. Past this size, only the application of level-$\left( k  \right)$ transition rules and the occurrence of level-$\left( k  \right)$ actual errors can cause $\mathcal{T}_{\left( k  \right)}$ to grow in size. 

We now consider the effect of level-$\left( k+1 \right)$ transition rules. The level-$\left( k+1 \right)$ colonies surrounding $\mathcal{B}'$ have valid level-$(k+1)$ syndromes reporting the vacuum. Indeed,  property \ref{lemma_causal_region} ensures that the causal region of a level-$k$ causally-linked cluster is contained within a box of size $2Q^{k'} \times 2Q^{k'} \times \left( 2  \textnormal{diam} \left( G_{\mathcal{A}} \right) + 3 \right) U^{k'}$. Since level-$k$ actual errors are $(aQ^k,aQ^k,bU^k)$-separated, there can be at most 1 level-$k$ actual error by bin of $b$ level-$k$ working periods. Using property \ref{lemma_valid_syndrome} for $k$ and the value of $f_n = \frac{4 \left( 3\textnormal{diam} \left( G_{\mathcal{A}} \right) + 1 \right) Q + 1}{b}$, every bin of $b$ level-$k$ working period reports the vacuum, ensuring that the level-$(k+1)$ corresponding syndrome reports the vacuum. 

Suppose that the level-$\left( k + 1 \right)$ syndromes are all valid within $\mathcal{B}'$. By inspecting every possible case (see appendix \ref{AppendixB}), it is clear that $\mathcal{T}_{\left( k \right)}$ will be kept inside $\mathcal{B}'$.

On the other hand, if some of the level-$\left( k + 1 \right)$ syndromes in $\mathcal{B} '$ are invalid, it is possible that the application of a level-$\left( k + 1 \right)$ transition rule brings $\mathcal{T}_{\left( k  \right)}$ outside of $\mathcal{B}'$. Suppose first that the total charge of the anyons part of $\tilde{\mathcal{C}} \left( E \right)$ in $C^{k+1}_{\bm \rho}\in \mathcal{B}'$ is non-trivial, but that the corresponding reported syndrome $s_{k+1,c}$ is the vacuum. This case cannot cause problems, since by design, the level-$\left( k + 1 \right)$ transition rules will not move a charge out of $C^{k+1}_{\bm \rho}$ when the syndrome is the vacuum. Suppose, however, that the reported syndrome is non-trivial. The transition rule that will be applied depends on the syndromes reported by the neighbouring colonies, $s_{k+1,n}$. If the total charge of the anyons part of $\tilde{\mathcal{C}} \left( E \right)$ in a neighbouring colony $C^{k+1}_{\bm \rho + (\alpha, \beta)}$ is non-trivial but its syndrome $s_{k+1,n}$ reports the vacuum, then it is possible that the application of a level-$\left( k + 1 \right)$ transition rule causes $\mathcal{T}_{\left( k \right)}$ to grow outside of $\mathcal{B}'$. It is in fact the only potentially problematic case.

We thus seek values for $b$ ensuring that if it is possible for the level-$(k+1)$ colony $C^{k+1}_{\bm \rho}$ to report a non-trivial syndrome $s_{k+1,c} \left( \bm \rho, \tau \right)$, then all the syndromes $s_{k+1,n} \left( \bm \rho + (\alpha, \beta), \tau \right)$ of neighboring colonies containing a nontrivial charge reports it (with $\break \alpha, \beta \in \{-1, 0, 1 \}$). Suppose that a non-trivial syndrome gets detected with $T$ level-$k$ working periods left to the current level-$(k+1)$ working period in the colony $C^{k+1}_{\bm \rho}$. The latest time after which a non-trivial charge may get detected in a neighbourhood level-$\left( k+1 \right)$ colony centre is $7 \left( b+2 \right)U^{k} + 4 \left( 3  \text{diam} \left( G_{\mathcal{A}} \right) + 1 \right)Q U^{k}$ time steps after that, by using the maximal temporal extension of a level-$(k+1)$ actual error and the proof of lemma \ref{lemma_faithful}. Using property \ref{lemma_valid_syndrome} for $k$, we find that after that point in time, the syndromes for the following bins of $bU^{k}$ time steps are valid. We thus require that if $ f_{c} b \leq \lceil \frac{T}{b} \rceil $, then $f_{n} b \leq \lfloor \frac{T-7\left( b+2 \right) - 4 \left( 3 \textnormal{ diam} \left( G_{\mathcal{A}} \right) + 1 \right) Q-1 }{b} \rfloor$. In other words, we require that if there is enough time left on the current level-$\left( k + 1 \right)$ working period for a syndrome $s_{k+1,c} $ to report a non-trivial charge (valid or not), then there is also enough time left to ensure that the syndromes $s_{k+1,n}$ of the neighbouring colonies are valid. Using the numerical values for $f_{c}$ and $f_{n}$, we find that choosing

\begin{equation} 
\begin{split}
b \geq b_{0} & \equiv 4 \left( 3 \text{diam} \left( G_{\mathcal{A}} \right) + 1 \right) Q + 5 \\
+& 2 \left( 4 \left( 3 \text{diam} \left( G_{\mathcal{A}} \right) + 1 \right) ^{2} Q^{2} + 11 \left( 3 \text{diam} \left( G_{\mathcal{A}} \right) + 1 \right) Q + 7 \right)^{\frac{1}{2}}
  	\label{eqn_b_restriction}
\end{split}
\end{equation}
ensures that it is always the case.

\end{proof}

\begin{lemma}
Suppose that our algorithm is $k$-local, $(k+1)$-faithful and $k$-successful. If $\mathcal{A}$ is non-cylic, then our algorithm is $\left( k + 1 \right)$-successful.
\label{lemma_successful}
\end{lemma}

\begin{proof}
Consider the active transition rule $M^{\left( \alpha, \beta \right),k+1}_{\bm \rho}$ applied at the end of a level-$(k+1)$ working period labelled by $\tau$ and moving a non-trivial anyon from colony $C^{k+1}_{\bm \rho}$ towards $C^{k+1}_{\bm \rho + (\alpha, \beta) }$. Suppose first that no actual error happens during the level-$(k+1)$ working periods $\tau$ and $\tau+1$.  Using $(k+1)$-faithfulness any anyon part of $\mathcal{C}^{(k+1)}$ is in the physical site at the centre of the level-$(k+1)$ colony for any time $t' \geq \tau U^{k+1} + (1-f_{c})bU^{k}$. It is thus clear that $M$ will move the anyon towards a neighbouring level-$(k+1)$ colony without interacting with any other non-trivial anyon. Furthermore, for the working period $\tau + 1$, the level-$(k+1)$ syndrome of $C^{k+1}_{\bm \rho}$ will be valid, since by using property \ref{lemma_valid_syndrome} for $k$, all the bins of $b$ level-$k$ working periods will report valid syndromes. Since no actual error nor anyon is present in $C^{k+1}_{\bm \rho}$ for the working period $\tau +1$, the reported charge will be the vacuum.

Consider next the case where actual errors of level-$k$ and lower happens on $C^{k+1}_{\bm \rho}$ during the working periods $\tau$ and $\tau+1$. Suppose first that the valid level-$(k+1)$ syndrome reports an abelian charge $a$. Since the anyon displaced at the end of the working period $\tau$ is abelian, it cannot interact with any actual error whose causal region intersect its path. The total charge in $C^{k+1}_{\bm \rho}$ after this process is then clearly the vacuum, since the net effect of the level-(k+1) transition rule is to displace an anyon of topological charge $a$ from $C^{k+1}_{\bm \rho}$ towards $C^{k+1}_{\rm \rho + (\alpha, \beta)}$. As argued previously, the level-$(k+1)$ syndrome of $C^{k+1}_{\bm \rho}$ for the working period $\tau + 1$ will report the vacuum as well.

\begin{figure}[!h]
\centering
\includegraphics[width=0.8 \textwidth]{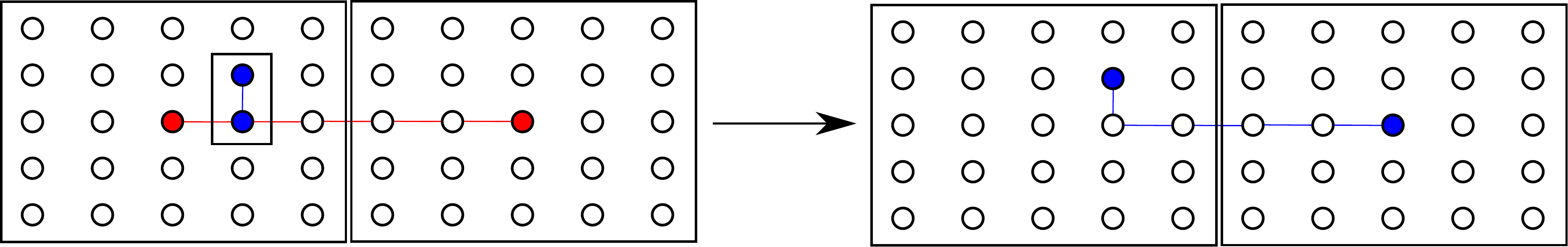}
\caption{An example of an actual level-$0$  error (blue line) occuring during the application of a level-$1$ transition rule (red line). In this example, the application of the transition rule is not successful, as after its application, anyons part of the level-$1$ actual error are still contained inside the colony to the left. For simplicity, $Q=5$ here, and such a process could happen with Ising anyons for example, where the red dots represent $\sigma$ anyons and the blue dots represent $\psi$ fermions.}
\label{fig_error_growth}
\end{figure}

Suppose, on the other hand, that the valid level-$(k+1)$ syndrome reports a non-abelian charge $a$. If the site at the centre of $C^{k+1}_{\bm \rho}$ does not contain an anyon of charge $a$ when the transition rule is applied, or if the moving anyon crosses the causal region of some actual errors, the application of $M^{\left( \alpha, \beta \right),k+1}_{\bm \rho}$ may be unsuccessful, as illustrated by figure \ref{fig_error_growth}.

%Consider first the case where the site at the center of $C^{k+1}_{(x,y)}$ does not contain an anyon with charge $a$, and that the path of the anyon being displaced by $M^{\left( \alpha, \beta \right),k+1}_{x,y}$ does not cross the causal region of any actual error. In this case, a pair of anyons of charges $a$ and $\bar a$ will be created from the vacuum, with the anyon of charge $\bar a$ being put in the center of $C^{k+1}_{(x,y)}$ and $a$ being be displaced to the the colony $C^{k+1}_{(x + \alpha,y + \beta)}$. Since $a$ and $\bar a$ are non-abelian, the total charge of $C^{k+1}_{(x,y)}$ for the next working period $\tau+1$ will be non trivial with probability of $1-d_a^{-2}$. Note however that the level-$(k+1)$ syndrome of $C^{k+1}_{(x,y)}$ for the working period $\tau+1$ will be valid, using the same argument as for property \ref{lemma_valid_syndrome}. 

Consider the case where the anyon of charge $a$ being moved by $M^{\left( \alpha, \beta \right),k+1}_{\bm \rho}$ interacts with a level-$k$ actual error $F$. By $k$-locality, the causal region of $F$, up to the time at which it interacts with the anyon $a$, is contained inside a group of $2 \times 2$ level-$k$ colonies lying on a corridor of $C^{k+1}_{\bm \rho}$ or $C^{k+1}_{\bm \rho + (\alpha,\beta)}$.

Suppose first that $F$ is entirely contained within $C^{k+1}_{\bm \rho +(\alpha, \beta)}$. In this case, by using property \ref{lemma_valid_syndrome} for $k$, the level-$k$ colonies containing $F$ will report valid level-$k$ syndromes once the actual error $F$ is over. By $k$-successfulness and by inspection of the transition rules, we find that after at most $2 \textnormal{diam} (G_{\mathcal A})$ level-$k$ working periods, all the anyons caused by $F$ will be contained in a single level-$k$ colony on a corridor of $C^{k+1}_{\bm \rho + (\alpha,\beta)}$. As argued in the proof of lemma \ref{lemma_faithful}, the $k$-trajectory of $\mathcal{C}^{(k+1)}$ in the colony $C^{k+1}_{\bm \rho +(\alpha,\beta)}$ is then confined inside of it, until the application of level-$(k+1)$ or larger transition rules. Furthermore, all the anyons part of $\mathcal{C}^{(k+1)}$ in $C^{k+1}_{\bm \rho + (\alpha,\beta)}$ will get to the centre of the level-$(k+1)$ colony, using $(k+1)$-faithfulness. Hence, the presence of $F$ does not affect the total charge of the anyons part of $\mathcal{C}^{(k+1)}$ in $C^{k+1}_{\bm \rho}$, nor can it affects its level-$(k+1)$ syndrome for working period $\tau + 1$.

Next suppose that $F$ is entirely contained within $C^{k+1}_{\bm \rho}$. As in the previous case, all the anyons part of $\mathcal{C}^{(k+1)}$ inside $C^{k+1}_{\bm \rho}$ will get to its centre fast enough to ensure that the level-$(k+1)$ syndrome of $C^{k+1}_{\bm \rho}$ for the next working period will be valid, using the same argument as for property \ref{lemma_valid_syndrome}. The state of the anyons can diagrammatically be written as a linear combination of terms of the form shown on the left-hand side of the following equation:
\begin{equation}
	\parbox{0.88\textwidth}{
		\includegraphics[scale=0.65]{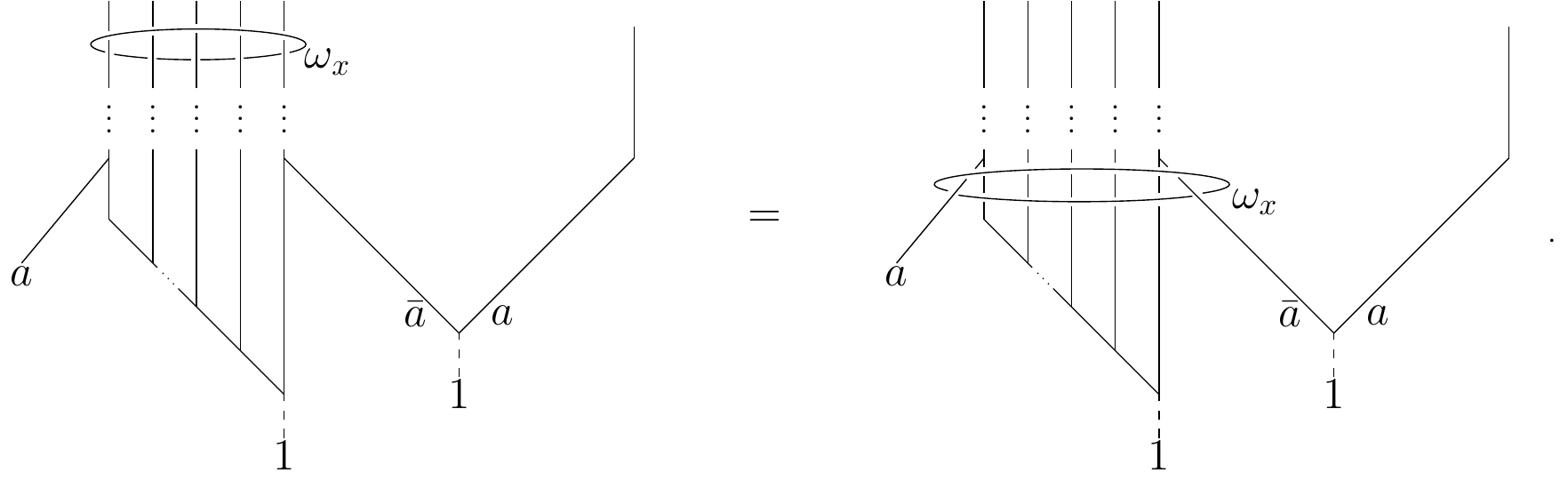}
	}
	\label{fig_meas_loop1}
\end{equation}

The charge reported at the centre of $C^{k+1}_{\bm \rho}$ can be measured using the $\omega$-loop shown in the left-hand side of equation (\ref{fig_meas_loop1}). After the repeated use of the Yang-Baxter equation (\ref{eqn_yang_baxter}), one gets to the right-hand side of equation (\ref{fig_meas_loop1}). Since $\mathcal{A}$ is non-cyclic, the total charge is different than $a$, since it cannot be so for any term in the linear combination.

It is possible that the group of $2 \times 2$ level-$k$ colonies containing $F$ overlaps with both $C^{k+1}_{\bm \rho}$ and $C^{k+1}_{\bm \rho + (\alpha, \beta)}$. In this case, one can explicitly check using the transition rules that all the non-trivial anyons part of the causally-linked cluster of $F$ and in the block of $2 \times 2$ level-$k$ colonies containing $F$ will end up in a single level-$k$ colony in a corridor of either $C^{k+1}_{\bm \rho}$ or $C^{k+1}_{\bm \rho + (\alpha, \beta)}$. Both these cases have been considered above. 

It is now clear that even in the presence of actual errors of level-$k$ or lower (the above reasoning holds for actual errors of lower levels as well and is readily generalized to the case where many actual errors of various levels lower than $k+1$ happen), after applying $M^{\left( \alpha, \beta \right),k+1}_{\bm \rho}$ at most $ \textnormal{diam} \left( G_{\mathcal{A}} \right) - 1 $ times, the  charge left in $C^{k+1}_{\bm \rho}$ is abelian. At this point, the next application of the transition rule will be applied successfully. In the above analysis, we have assumed that the same level-$(k+1)$ transition rule (although based on different syndromes in $C^{k+1}_{\bm \rho}$) is applied at the various level-$(k+1)$ working periods. However this is not crucial and they need not be the same, \emph{i.e.} the non-trivial anyon may be moved towards different neighbouring level-$\left( k + 1 \right)$ colonies in $\{ C^{k+1}_{\langle \bm \rho \rangle} \}$ at different level-$(k+1)$ working periods without affecting the previous reasoning.

\end{proof}

\begin{theorem}
If $\mathcal{A}$ is non-cylic, then there exists a critical value $p_{c} > 0$ such that if $p + q < p_{c}$, for any number of time steps $T$ and any $\epsilon > 0$, there exists a linear system size $L = Q^{n} \in \mathcal O(\log \frac 1\epsilon)$ such that with probability of at least $1 - \epsilon$, the encoded quantum state can in principle be recovered after $T$ time steps.
\label{thm_threshold}
\end{theorem}

\begin{proof}
Let $a = 3$, $Q > 14 \left( a + 2 \right) + 7$ and $b > b_{0}$. Define $p_{c} = \frac{Q^{-4}b^{-4}}{4}$, and let $p + q < p_{c}$. Choose the linear size of the system $L = Q^{n}$, with $n$ such that $U^{n} \geq T$ and $\left( \frac{ \left( p + q \right)}{p_c} \right)^{2^{n}} \leq \epsilon$.

Since $p + q < p_{c}$, lemma \ref{lem_actual_error} guarantees that with probability $1$, any error is part of an actual error and thus the concept of causally-linked clusters is well-defined. Using lemma \ref{lemma_0_level}, the correction algorithm is $0$-local, $0$-faithful and $0$-successful. Using lemmas \ref{lemma_faithful} through \ref{lemma_successful}, we recursively find that the algorithm is $\left( n -1 \right)$-local, $\left( n - 1 \right)$-faithful and $\left( n - 1 \right)$-successful as well. It follows from property \ref{lemma_causal_region} and the separation between actual errors of the same level that in order to perform a non-trivial operation on the encoded subspace, at least one causally-linked cluster of level-$n$ or larger is required. By definition, a level-$n$ causally-linked cluster requires the presence of at least one level-$n$ actual error. Using lemma \ref{lem_error_rate}, the probability of having at least one level-$n$ or larger actual error having non-empty intersection with a box in space time of size $Q^{n} \times Q^{n} \times U^{n}$  is at most $\left( \frac{ \left( p + q \right)}{ p_{c} } \right)^{2^{n}}$. We thus find if the transition rules continue to operate in an error-free way after time $T$, the system will return to its initial ground state with probability at least $1 - \epsilon$.
\end{proof}

Setting $a = 3$, $Q = 78$, and $b = 9(3D + 1)Q$ respects all the conditions under which the above theorem holds and leads to a threshold lower-bound of $2,7 \times 10^{-20} \times (3D +1)^{-4}$. In particular for a system of Ising anyons where $D=2$, we obtain a threshold lower-bound of $1,1 \times 10^{-23}$.

\section{Numerical Simulations of a System of Ising Anyons}
\label{sec_numerical}

The numerical value of the threshold found in the preceding section are extremely low and should be understood as a proof of existence. In practice, the threshold is likely to be much higher, since the proof presented here consider the worst-case scenario, which may differ significantly from the average case. To verify this claim, numerical simulations for a system of Ising anyons living on a torus were performed. The superselection sectors of the theory are $\{ 1, \epsilon, \sigma \}$ and the non-trivial fusion rules are given by

\begin{eqnarray}
\epsilon \times \epsilon = 1, \\
\sigma \times \epsilon = \sigma, \\
\sigma \times \sigma = 1 + \epsilon.
\end{eqnarray}

A complete description of Ising anyons can be found in~\cite{Kitaev_06,Rowell_09}. Ising anyons are non-cyclic, with a graph diameter of $2$. Moreover, the evolution of a system of Ising anyons restricted to topological operations ({\it i.e.} splitting, fusion and braiding) can be efficiently simulated classically. This can be achieved by using the spinor representation of the braid group acting on $2n$ $\sigma$ anyons, described by the $SO(2n)$ group ~\cite{Bravyi_06,Nayak_96,Brell_14} and by using the algorithm presented in ~\cite{Aaronson_04}.

We have performed numerical simulations of a system of Ising anyons on a $L \times L$ torus. The various values of $L$ considered where $L = Q^{n}$, for $Q = 3$ and $n \in \{ 1,2,3,4\}$. An optimal renormalization time $U=b^{2}$ was empirically found to be $U = 49$ for the parameters presented here. The measurement error rate $q$ is set equal to the charge creation error rate $p$. The fractions $f_c$ and $f_n$ were set to $0.8$ and $0.2$ respectively, and were also empirically found to work best. Using lower measurement error rates $q$ gives very similar results.

\begin{figure}[h!]
	\begin{center}
		\includegraphics[width=.85 \textwidth]{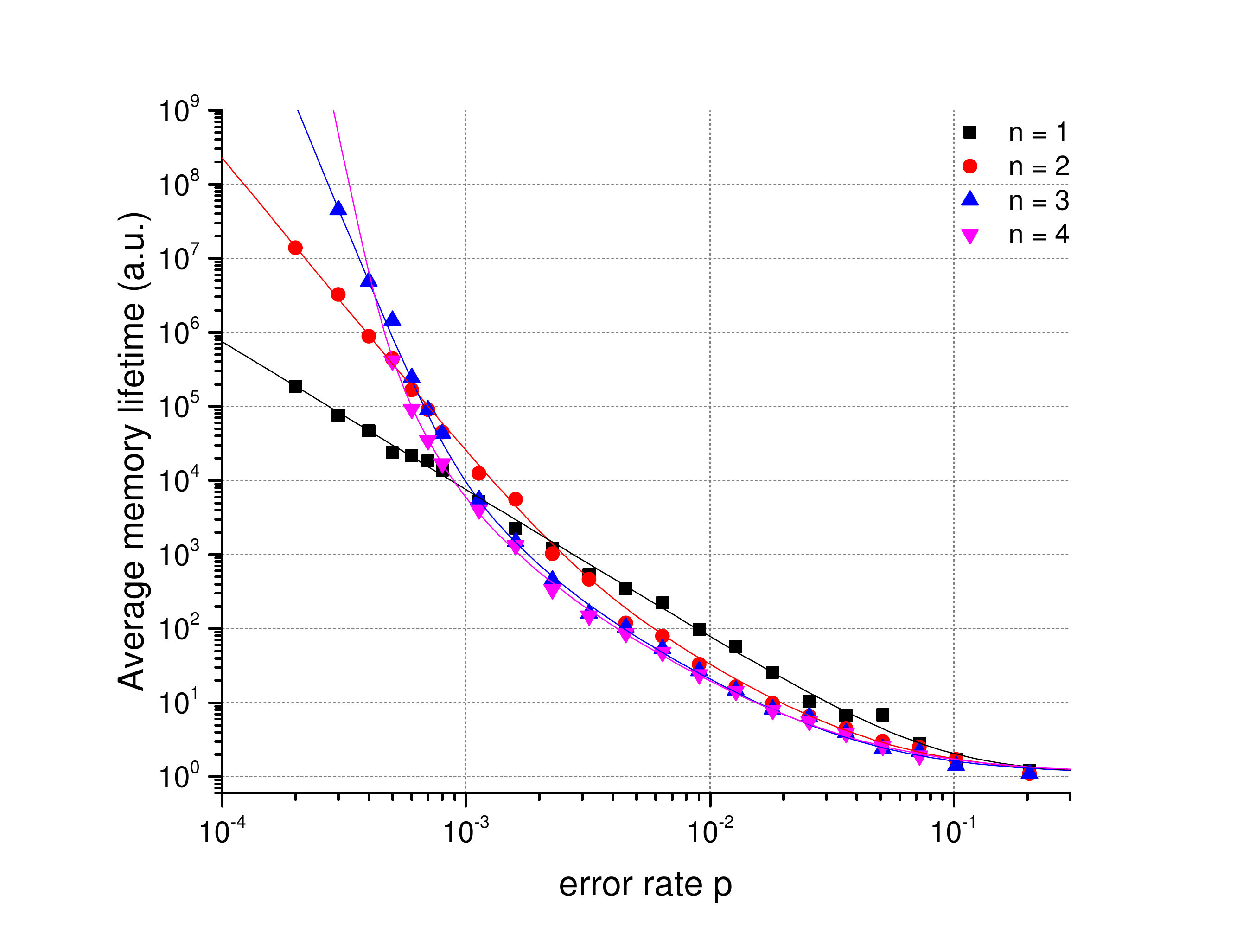}
		\caption{Results of numerical simulations for a system of Ising anyons on a torus of linear size $3^{n} \times 3^{n}$ for values of $n = 1, 2, 3$ and $4$, using a renormalization time $U = b^{2} = 49$. The thresholds $f_{c}$ and $f_{n}$ were set to $0.8$ and $0.2$, respectively. Each point represents the average lifetime of 60 instances for the given error rate. The measurement error rate $q$ is set equal to the charge error rate $p$. The solid lines are polynomial fittings of the argument $\left( \frac{1}{p \left( 1 - p \right)} \right)$ of power $2^{n}$.}
		\label{fig_numerical_simul}
	\end{center}
\end{figure}

At every time step, after applying both the error model and the local transition rules (on the renormalized levels if at suitable time), the decoding algorithm introduced in~\cite{Brell_14} was used to verify whether the information could in principle be retrieved. This algorithm acts as a two-step perfect matching algorithm (performed using the blossomV software~\cite{Kolmogorov_09}), first matching pairs of $\sigma$ anyons which may fuse either to the vacuum or give an $\epsilon$ fermion, and then matching $\epsilon$ fermions in pairs, thus always returning to the vacuum. Note that during this virtual verification procedure, measurements were performed perfectly. The number of time steps before this verification procedure reports that the information was corrupted is called the lifetime of this particular instance. Note that because it uses a heuristic, suboptimal verification procedure, this reported lifetime should be interpreted as a lower bound to the real lifetime. Average lifetime as a function of the error rate is shown in Figure~\ref{fig_numerical_simul}. Polynomial fittings of degree $2^{n}$ in the variable$ \left( \frac{1}{p \left( 1 - p \right)} \right)$ are shown as continuous lines in the plot as well. The asymptotic behaviour for low error rates of $\sim p^{-2^{n}}$ is in agreement with the bounds found in Sec.~\ref{bound_error_rate}, even though here $Q = 3$ only. Note that Figure~\ref{fig_numerical_simul} does not allow to determine the asymptotic region of low error rate for the case of $n=4$, but the simulation results are in agreement with an average lifetime $\sim p^{-2^{4}}$. Such a low value of $Q$ was used because of the exponential growth of the size of the lattice with the level-$n$, and larger values of $Q$ did not allow for simulation in a reasonable time.

While the  simulations were limited to values of $n$ too small to observe a clear threshold behaviour, they are consistent with a threshold in the $10^{-4}-10^{-3}$ range. In particular, the data and polynomial fits show that, for values of $p$ below $p_c \approx 3\times 10^{-4}$, the average memory lifetime increases with the lattice size. 

\section{Discussion}
\label{sec_discussion}

We have shown, both analytically and numerically, that an adaptation of a cellular automaton developed by G\'acs \cite{Gacs_86} and Harrington \cite{Harrington_04} can implement fault-tolerant quantum error correction for a system of non-cyclic modular anyons.  The $\sim 10^{-23}$ threshold lower-bound we found analytically is strikingly low. In contrast, Harrington \cite{Harrington_04} reports an analytical threshold lower bound $\sim 10^{-12}$ for Abelian anyons. This difference is explained by the possibility that a low-level error modifies the trajectory of a level-$k$ anyon and make it escapes the $2\times 2$ block of level-$k$ colonies containing its trajectory; a possibility not accounted for in Harrington's analysis. This possibility forces us to adapt a larger value of $Q=78$ compared to $Q=16$ used by Harrington. 

Despite this prohibitively low lower bound, numerical simulations suggest that a thermal anyon density $e^{-\frac{\Delta}{k_{b} T} } \sim 10^{-4} - 10^{-3}$, which corresponds to a temperature $T$ about an order of magnitude below the spectral gap $\Delta$, in conjunction with a relatively low measurement error rate (in the range $0.1 - 0.01 \% $) are sufficient. This is the first theoretical demonstration of fault-tolerant quantum error correction for non-abelian anyons which can tolerate faulty measurements; it is likely that other schemes will perform better.

There are indeed many avenues to improve upon our scheme. Most noteworthy, our algorithm does not take into account any specificity of the anyonic models, such as the fusion rules. Its sole objective is to fuse non-trivial anyons together, irrespective of the possibility that the fusion results in the vacuum. Moreover, the algorithm does not use the partial information it generates about the anyon world-lines. For instance, when displacing a level-$k$ excitation, the algorithm could make use of lower-level information to try to avoid linking the displaced excitation with lower-level errors. An algorithm exploiting this extra information could lead to significant quantitative improvements.

The most important simplification we made in our analysis is that the application of the renormalized transition rules can be performed in a single time step. While it is possible to displace an abelian anyon over an arbitrary large distance in constant time (by a parallel repetition of fissions and fusions teleporting the anyon), non-abelian anyons can only propagate at a finite velocity. This finite velocity requires modifying the transition rules to avoid conflicts between the displacement of a high-level charge and the application of lower-level transition rules. For instance, we could simply ``turn off'' the lower-level transition rules during the passage of a high-level excitation. These vacant transition rules can be seen as additional errors, so they would essentially have the effect of slightly increasing the error rate of lower levels. Thus, we believe that the finite velocity of non-abelian anyons should only quantitatively change our analysis and simulations, but leave our qualitative conclusions unchanged. 

The idea of encoding information on a surface of genus 1 or higher seems technically challenging. A more natural way to encode information using non-abelian anyons is to use the fusion space of the non-abelian species or to add punctures on a planar surface \cite{Beverland_14}. This last procedure is commonly used in the context of surface codes \cite{Freedman_2001}. While our algorithm does not directly apply to such settings, we believe that it should apply with only minor modifications. In particular, provided that these topological defects are kept at a distance $\geq Q^n$ from each other, there should exist a colony patchwork which avoid them in such a way that they do not interfere with the transition rules of level $n$ or less. A similar technique, perhaps using a time-dependent patchwork, could be applied in a setting where anyons are braided and fused to realize fault-tolerant anyonic quantum computation. 

The analysis we presented fails in the case of cyclic anyonic models, such as Fibonacci anyons. For these models, there is no guarantee that the transition rules can be applied successfully. This leads to the possibility that an actual error fails to  get corrected before the appearance of other actual errors of the same or larger level.  Since the probability that a particle and its anti-particle fuse to a non-trivial particle $\ell$ times in a row drops exponentially with $\ell$, we could imagine imposing a hard cutoff $\ell_{\rm max}$ on the number of fusion attempts. Our analysis would then go through by replacing $D=diam(G_{\mathcal A})$ by $\ell_{\rm max}$, leaving an error floor exponentially low in $\ell_{\rm max}$. The problem with this approach is that to achieve an exponential error suppression with the lattice size $L$, the cutoff $\ell_{\rm max}$ would need to scale with $L$, resulting in a ``drifting threshold''.  Hence, constructions that do not use renormalization ideas such as the one presented in \cite{Herold_14, Herold_15} may be necessary to study fault-tolerance for cyclic anyons.

\section{Acknowledgements}

We thank Guillaume Duclos-Cianci for valuable discussions as well as Steve Allen for technical assistance. GD was partially supported by the Fonds de recherche du Qu\'ebec - Nature et technologies. This work was partially funded by the Natural Sciences and Engineering Research Council of Canada and the Canadian Institute for Advanced Research. Computations were made on the supercomputer Mammouth from Universit\'e de Sherbrooke, managed by Calcul Qu\'ebec and Compute Canada. The operation of this supercomputer is funded by the Canada Foundation for innovation (CFI), the minist\`ere de l'\'Economie, de la science et de l'innovation du Qu\'ebec (MESI) and the Fonds de recherche du Qu\'ebec - Nature et technologies (FRQ-NT).

\bibliography{proof_sketch}
\bibliographystyle{unsrtnat}

\appendix

\section{Explicit description of the transition rules} \label{AppendixA}

Let $(\alpha, \beta)$ be a two-dimensional vector in  $\{ (1,0), (-1,0), (0,1), (0-1) \}$ and $l_q$ be the label of topological charge $q$ (with $l_1 = 0$ by convention). The operator $M_{\mathbf{r}}^{\left( \alpha, \beta \right) } \left( l_q \right)$ is operationally defined as follows: if there is a particle of topological charge $q$ at site $\mathbf{r}$, then displace that particle from site $\mathbf{r}$ to the site $\mathbf{r} + (\alpha, \beta)$ via the edge connecting the two sites; if no such particle is present, then first create a pair of particles of charges $q$ and $\bar{q}$ from the vacuum at site $\mathbf{r}$, and the particle with charge $q$ is then displaced to the site $\mathbf{r} + (\alpha, \beta)$ in the same way.

The level-$k$ operator $M_{\bm{\rho}}^{ \left( \alpha, \beta \right), k} (l_q)$  is defined by \newline $M_{\bm{\rho}}^{ \left( \alpha, \beta \right), k} (l_q) = \prod_{i = 0}^{Q^{k}-1} P_{\mathfrak{c} (\bm \rho) + i \left( \alpha, \beta \right) }^{q_i} M_{\mathfrak{c} (\bm \rho) + i \left( \alpha, \beta \right) }^{\left( \alpha, \beta \right) } \left( l_q \right)$, with $\mathfrak{c} (\bm \rho)$ the physical site at the centre of the colony $C^{k}_{\bm \rho}$. It is understood that the operators $M_{\mathfrak{c} (\bm \rho) + i \left( \alpha, \beta \right) }^{\left( \alpha, \beta \right) } \left( l_q \right)$ are applied sequentially, each one followed by the application of $P_{\mathfrak{c} (\bm \rho) + i \left( \alpha, \beta \right) }^{q_i}$, the projector of the total charge $q_i$ at site $\mathfrak{c} (\bm \rho) + i \left( \alpha, \beta \right)$. The measured charges $\{ q_i \}$ are determined probabilistically, using equation (\ref{eqn_meas_prob}).

The operator $M_{\bm \rho }^{ \left( \alpha, \beta \right), k}( q) $ can be understood as moving an anyon of topological charge $q$ from the physical site corresponding to centre of the level-$k$ colony lying at renormalized site $\bm \rho$ all the way to the physical site corresponding to the centre of $\bm \rho + (\alpha, \beta)$. If the charge being displaced encounters a non trivial anyon on its path, the total charge at that site is measured. If the resulting charge is $q$, then it simply continues on its path. If, however, the resulting charge is different than $q$, then a pair of charges $q$ and $\bar{q}$ are created from the vacuum, the anyon with charge $\bar{q}$ is put in the site while the charge $q$ keeps being displaced.

The level-$k$ transition rules used by our algorithm are explicitly stated below. Since all syndromes are measured at the same time $t = (\tau+1) U^k- \frac{1}{2}$, with $\tau \in \mathbb{N}$, it is omitted in the syndrome notation. The notation $\bm \rho = ( \mathpzc{x}, \mathpzc{y} )$ is used below.

\begin{itemize}

\item { (West border) IF $\mathpzc{x}$ modulo $Q$ = 0 THEN \newline
if $s_{k,c} (\bm \rho)$ = 0, do nothing; \newline
else if $s_{k,n} (\bm \rho + (-1,1)) \neq 0$ , apply $M_{\bm \rho}^{\left( -1,0 \right),k}  \left( s_{k,c} (\bm \rho) \right)$; \newline
else if $s_{k,n} (\bm \rho + (-1,0)) \neq 0$, apply $M_{\bm \rho}^{\left( -1,0 \right),k} \left( s_{k,c} (\bm \rho) \right)$; \newline
else if $s_{k,n} (\bm \rho + (-1,-1)) \neq 0$, apply $M_{\bm \rho}^{\left( -1,0 \right),k} \left( s_{k,c} (\bm \rho) \right)$; \newline
else continue below.}

\item { (South border) IF $\mathpzc{y}$ modulo $Q$ = 0 THEN \newline
if $s_{k,c} (\bm \rho)$ = 0, do nothing; \newline
else if $s_{k,n} (\bm \rho + (-1,-1)) \neq 0$, apply $M_{\bm \rho}^{\left( 0,-1 \right),k} \left( s_{k,c} (\bm \rho) \right)$; \newline
else if $s_{k,n} (\bm \rho + (0, -1)) \neq 0$, apply $M_{\bm \rho}^{\left( 0,-1 \right),k} \left( s_{k,c} (\bm \rho) \right)$; \newline
else if $s_{k,n} (\bm \rho + (+1,-1)) \neq 0$, apply $M_{\bm \rho}^{\left( 0,-1 \right),k} \left( s_{k,c} (\bm \rho) \right)$; \newline
else continue below.}

\item { (South West quadrant) IF $\mathpzc{x}$ modulo $Q$ $< \lfloor \frac{Q}{2} \rfloor$ AND $\mathpzc{y}$ modulo $Q$ $< \lfloor \frac{Q}{2} \rfloor$ THEN \newline
if $s_{\bm \rho}^{k,c}$ = 0, do nothing; \newline
else if $s_{k,n} (\bm \rho + (0,-1)) \neq 0$, do nothing; \newline
else if $s_{k,n} (\bm \rho + (-1,0)) \neq 0$, do nothing; \newline
else if $s_{k,n} (\bm \rho + (-1,-1)) \neq 0$, do nothing; \newline
else if $s_{k,n} (\bm \rho + (0,1)) \neq 0$, apply $M_{\bm \rho}^{\left( 0,+1 \right),k} \left( s_{k,c} ({\bm \rho}) \right)$; \newline
else if $s_{k,n} (\bm \rho + (-1,1)) \neq 0$, apply $M_{\bm \rho}^{\left( 0,+1 \right),k} \left( s_{k,c} ({\bm \rho}) \right)$; \newline
else if $s_{k,n} (\bm \rho + (1,0)) \neq 0$, apply $M_{\bm \rho}^{\left( +1,0 \right),k} \left( s_{k,c} ({\bm \rho}) \right)$; \newline
else if $s_{k,n} (\bm \rho + (1,-1)) \neq 0$, apply $M_{\bm \rho}^{\left( +1,0 \right),k} \left( s_{k,c} ({\bm \rho}) \right)$; \newline
else, apply $M_{\bm \rho}^{\left( 0,+1 \right),k} \left( s_{k,c} ({\bm \rho}) \right) $.}

\item { (West corridor) IF $\mathpzc{x}$ modulo $Q$ $< \lfloor \frac{Q}{2} \rfloor$ AND $\mathpzc{y}$ modulo $Q$  $= \lfloor \frac{Q}{2} \rfloor$ THEN \newline
IF $s_{k,c} ({\bm \rho})$ = 0, do nothing; \newline
else if $s_{k,n} (\bm \rho + (0,-1)) \neq 0$, do nothing; \newline
else if $s_{k,n} (\bm \rho + (-1,0)) \neq 0$, do nothing; \newline
else if $s_{k,n} (\bm \rho + (0,1)) \neq 0$, do nothing; \newline
%else if $s_{k,n} (\bm \rho + (1,0)) \neq 0$, apply $M_{\mathpzc{x},y}^{\left( +1,0 \right),k} \left( s_{k,c} ({\bm \rho}) \right)$; \newline
else if $s_{k,n} (\bm \rho + (-1,-1)) \neq 0$, do nothing; \newline
else if $s_{k,n} (\bm \rho + (-1,1))  \neq 0$, do nothing; \newline
else, apply $M_{\bm \rho}^{\left( +1,0 \right),k} \left( s_{k,c} ({\bm \rho}) \right)$.}

\item { (North West quadrant) IF $\mathpzc{x}$ modulo $Q$ $< \lfloor \frac{Q}{2} \rfloor$ AND $\mathpzc{y}$ modulo $Q$ $> \lfloor \frac{Q}{2} \rfloor$ THEN \newline
IF $s_{\bm \rho}^{k,c}$ = 0, do nothing; \newline
else if $s_{k,n} (\bm \rho + (-1,0)) \neq 0$, do nothing; \newline
else if $s_{k,n} (\bm \rho + (0,1)) \neq 0$, do nothing; \newline
else if $s_{k,n} (\bm \rho + (-1,1)) \neq 0$, do nothing; \newline
else if $s_{k,n} (\bm \rho + (1,0)) \neq 0$, apply $M_{\bm \rho}^{\left( +1,0 \right),k} \left( s_{k,c} ({\bm \rho}) \right)$; \newline
else if $s_{k,n} (\bm \rho + (1,1)) \neq 0$, apply $M_{\mathpzc{x},y}^{\left( +1,0 \right),k} \left( s_{k,c} ({\bm \rho}) \right)$; \newline
else if $s_{k,n} (\bm \rho + (0,-1)) \neq 0$, apply $M_{\bm \rho}^{\left( 0,-1 \right),k} \left( s_{k,c} ({\bm \rho}) \right)$; \newline
else if $s_{k,n} (\bm \rho + (-1,-1)) \neq 0$, apply $M_{\mathpzc{x},y}^{\left( 0,-1 \right),k} \left( s_{k,c} ({\bm \rho}) \right)$; \newline
else, apply $M_{\bm \rho}^{\left( +1,0 \right),k} \left( s_{k,c} (\bm \rho) \right)$.}

\item { (North corridor) IF $\mathpzc{x}$ modulo $Q$ $= \lfloor \frac{Q}{2} \rfloor$ AND $\mathpzc{y}$ modulo $Q$ $> \lfloor \frac{Q}{2} \rfloor$ THEN \newline
IF $s_{k,c} ({\bm \rho})$ = 0, do nothing; \newline
else if $s_{k,n} (\bm \rho + (-1,0)) \neq 0$, do nothing; \newline
else if $s_{k,n} (\bm \rho + (0,1)) \neq 0$, do nothing; \newline
else if $s_{k,n} (\bm \rho + (1,0)) \neq 0$, do nothing; \newline
%else if $s_{k,n} (\bm \rho + (0,-1)) \neq 0$, apply $M_{\bm \rho}^{\left( 0,-1 \right),k} \left( s_{k,c} ({\bm \rho}) \right)$; \newline
else if $s_{k,n} (\bm \rho + (-1,1)) \neq 0$, do nothing; \newline
else if $s_{k,n} (\bm \rho + (1,1)) \neq 0$, do nothing; \newline
else, apply $M_{\bm \rho}^{\left( 0,-1 \right),k} \left( s_{k,c} ({\bm \rho}) \right)$.}

\item { (North East quadrant) IF $\mathpzc{x}$ modulo $Q$ $> \lfloor \frac{Q}{2} \rfloor$ AND $\mathpzc{y}$ modulo $Q$ $> \lfloor \frac{Q}{2} \rfloor$ THEN \newline
IF $s_{k,c} ({\bm \rho})$ = 0, do nothing; \newline
else if $s_{k,n} (\bm \rho + (0,1)) \neq 0$, do nothing; \newline
else if $s_{k,n} (\bm \rho + (1,0)) \neq 0$, do nothing; \newline
else if $s_{k,n} (\bm \rho + (1,1)) \neq 0$, do nothing; \newline
else if $s_{k,n} (\bm \rho + (0,-1)) \neq 0$, apply $M_{\bm \rho}^{\left( 0,-1 \right),k} \left( s_{k,c} ({\bm \rho}) \right)$; \newline
else if $s_{k,n} (\bm \rho + (1,-1)) \neq 0$, apply $M_{\bm \rho}^{\left( 0,-1 \right),k} \left( s_{k,c} ({\bm \rho}) \right)$; \newline
else if $s_{k,n} (\bm \rho + (-1,0)) \neq 0$, apply $M_{\bm \rho}^{\left( -1,0 \right),k} \left( s_{k,c} ({\bm \rho}) \right)$; \newline
else if $s_{k,n} (\bm \rho + (-1,1)) \neq 0$, apply $M_{\bm \rho}^{\left( -1,0 \right),k} \left( s_{k,c} ({\bm \rho}) \right)$; \newline
else, apply $M_{\bm \rho}^{\left( 0,-1 \right),k} \left( s_{k,c} ({\bm \rho}) \right)$.}

\item { (East corridor) IF $\mathpzc{x}$ modulo $Q$ $> \lfloor \frac{Q}{2} \rfloor$ AND $\mathpzc{y}$ modulo $Q$ $= \lfloor \frac{Q}{2} \rfloor$ THEN \newline
IF $s_{k,c} ({\bm \rho})$ = 0, do nothing; \newline
else if $s_{k,n} (\bm \rho + (0,1)) \neq 0$, do nothing; \newline
else if $s_{k,n} (\bm \rho + (1,0)) \neq 0$, do nothing; \newline
else if $s_{k,n} (\bm \rho + (0,-1)) \neq 0$, do nothing; \newline
%else if $s_{k,n} (\bm \rho + (-1,0)) \neq 0$, apply $M_{\bm \rho}^{\left( -1,0 \right),k} \left( s_{k,c} ({\bm \rho}) \right)$; \newline
else if $s_{k,n} (\bm \rho + (1,1)) \neq 0$, do nothing; \newline
else if $s_{k,n} (\bm \rho + (1,-1)) \neq 0$, do nothing; \newline
else, apply $M_{\bm \rho}^{\left( -1,0 \right),k} \left( s_{k,c} ({\bm \rho}) \right)$.}

\item { (South East quadrant) IF $\mathpzc{x}$ modulo $Q$ $> \lfloor \frac{Q}{2} \rfloor$ AND $\mathpzc{y}$ modulo $Q$ $< \lfloor \frac{Q}{2} \rfloor$ THEN \newline
IF $s_{k,c} ({\bm \rho})$ = 0, do nothing; \newline
else if $s_{k,n} (\bm \rho + (1,0)) \neq 0$, do nothing; \newline
else if $s_{k,n} (\bm \rho + (0,-1)) \neq 0$, do nothing; \newline
else if $s_{k,n} (\bm \rho + (1,-1)) \neq 0$, do nothing; \newline
else if $s_{k,n} (\bm \rho + (-1,0)) \neq 0$, apply $M_{\bm \rho}^{\left( -1,0 \right),k} \left( s_{k,c} ({\bm \rho}) \right)$; \newline
else if $s_{k,n} (\bm \rho + (-1,-1)) \neq 0$, apply $M_{\bm \rho}^{\left( -1,0 \right),k} \left( s_{k,c} ({\bm \rho}) \right)$; \newline
else if $s_{k,n} (\bm \rho + (0,1)) \neq 0$, apply $M_{\bm \rho}^{\left( 0,+1 \right),k} \left( s_{k,c} ({\bm \rho}) \right)$; \newline
else if $s_{k,n} (\bm \rho + (1,1)) \neq 0$, apply $M_{\bm \rho}^{\left( 0,+1 \right),k} \left( s_{k,c} ({\bm \rho}) \right)$; \newline
else, apply $M_{\bm \rho}^{\left( -1,0 \right),k} \left( s_{k,c} ({\bm \rho}) \right)$.}

\item { (South corridor) IF $\mathpzc{x}$ modulo $Q$ $= \lfloor \frac{Q}{2} \rfloor$ AND $\mathpzc{y}$ modulo $Q$ $< \lfloor \frac{Q}{2} \rfloor$ THEN \newline
IF $s_{k,c} ({\bm \rho})$ = 0, do nothing; \newline
else if $s_{k,n} (\bm \rho + (1,0)) \neq 0$, do nothing; \newline
else if $s_{k,n} (\bm \rho + (0,-1)) \neq 0$, do nothing; \newline
else if $s_{k,n} (\bm \rho + (-1,0)) \neq 0$, do nothing; \newline
%else if $s_{k,n} (\bm \rho + (0,1)) \neq 0$, apply $M_{\bm \rho}^{\left( 0,+1 \right),k} \left( s_{k,c} ({\bm \rho}) \right)$; \newline
else if $s_{k,n} (\bm \rho + (1,-1)) \neq 0$, do nothing; \newline
else if $s_{k,n} (\bm \rho + (-1,-1)) \neq 0$, do nothing; \newline
else, apply $M_{\bm \rho}^{\left( 0,+1 \right),k} \left( s_{k,c} ({\bm \rho}) \right)$.}

\item { (Colony centre) IF $\mathpzc{x}$ modulo $Q$ $= \lfloor \frac{Q}{2} \rfloor$ AND $\mathpzc{y}$ modulo $Q$ $= \lfloor \frac{Q}{2} \rfloor$ THEN;
	update colony syndromes $s_{k+1,c} ({\bm \rho})$ and $s_{k+1,n} ({\bm \rho})$.}
	
\end{itemize}

%\begin{figure}[h]
%\centering
%\includegraphics[width=0.80 \textwidth]{colony_classification}
%\caption{The structure of a single colony of $9 \times 9$ sites. The four quadrants are denoted by the cardinal points SW, NW, NE and SW; while the four central corridors are identified with the cardinal points W, N, E and S. The West and South borders are also identified, and the colony center is depicted in red.}
%\label{fig_colony_description}
%\end{figure}

\section{Level-$0$ actual errors getting corrected within 2 time steps} \label{AppendixB}

In this section, it is explicitly shown that every actual level-$0$ error gets corrected within 2 time steps. Figures~\ref{fig_ferm_error} through~\ref{fig_meas_border} list all possible cases (up to rotations) of actual level-$0$ errors. In these figures, blue dots represent non-trivial topological charges, without distinction for their specific charges. ``M" represents a measurement step, with the occupied site corresponding to the reported syndromes, and not necessarily actual charges, while ``U" denotes the application of the local rules, based on the previous measurements. The  arrows represents the transition rules described above. To simplify notation, we do not distinguish between different topological charges since the transition rules are oblivious to them. In the case where a fusion happens after the first correction step, the corresponding blue dot represent any possible fusion result giving a non-trivial charge. In the case where the fusion result may yield a trivial charge as well, the rest of the process can simply be replaced by the appropriate case. The sites are assumed to be in the south-west quadrant of a colony on the following figures, but errors get corrected in a similar fashion in every quadrant of the colony, since careful inspection of the transition rules reveal that they are symmetric under rotations of $\frac{\pi}{4}$ around the colony centre.

Note that the analysis of these cases also works at the $k^{\rm th}$ level of renormalization if the algorithm is $k$-local, $k$-faithful and $k$-successful. In this case, the sites correspond to level-$k$ colonies. The $t = 0$ step represents the possible total charge of the anyons of the level-$k$ causally-linked cluster of a level-$k$ actual error $E$ at the end of the last level-$k$ working period into which errors part of $E$ are supported. If the total charge is non-trivial in the corresponding level-$k$ colonies, the $t = 0.5$ step depicts the corresponding valid level-$k$ syndromes. If no non-trivial level-$k$ charge is present in the $t = 0$ step, the $t = 0.5$ step represents the reported syndromes, which are invalid. The step at $t = 1$ represents the state of the system after the successful application of a round of active level-$k$ transition rules in the case where syndromes were valid, or the resulting state of the system after the application of transition rules if the syndromes were not valid. Steps denoted by $t = 1.5$ and $t = 2$ are similarly interpreted, except that the reported syndromes at $t = 1.5$ are always valid.

\begin{figure}[!h]
\centering
\includegraphics[width=0.9 \textwidth]{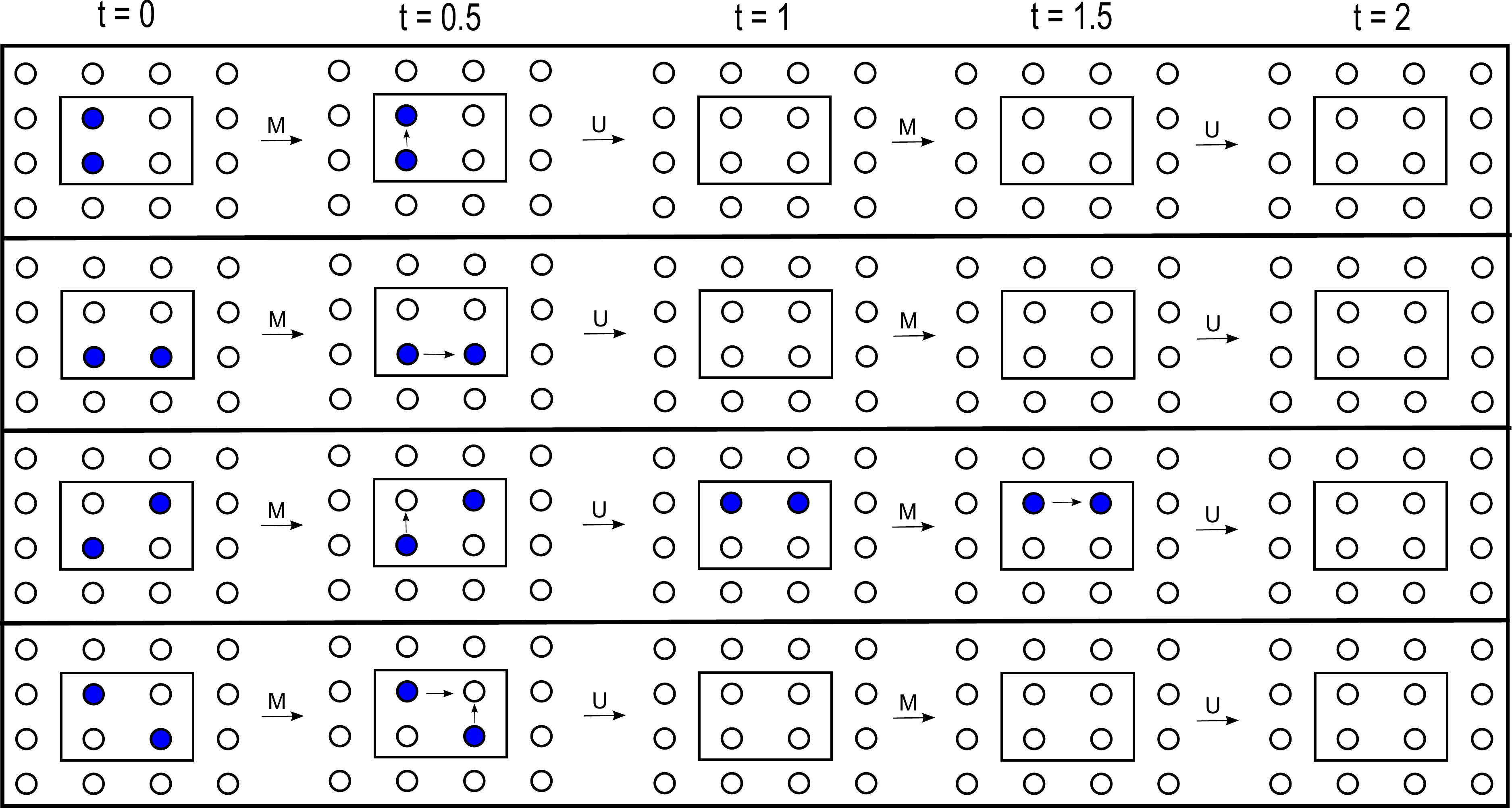}
\caption{Actual level-$0$ errors containing 2  non-trivial topological charges get corrected within 2 time steps. Note that the specific charges do not matter, as the total charge must be the vacuum.}
\label{fig_ferm_error}
\end{figure}

\begin{figure}[!h]
\centering
\includegraphics[width=0.9 \textwidth]{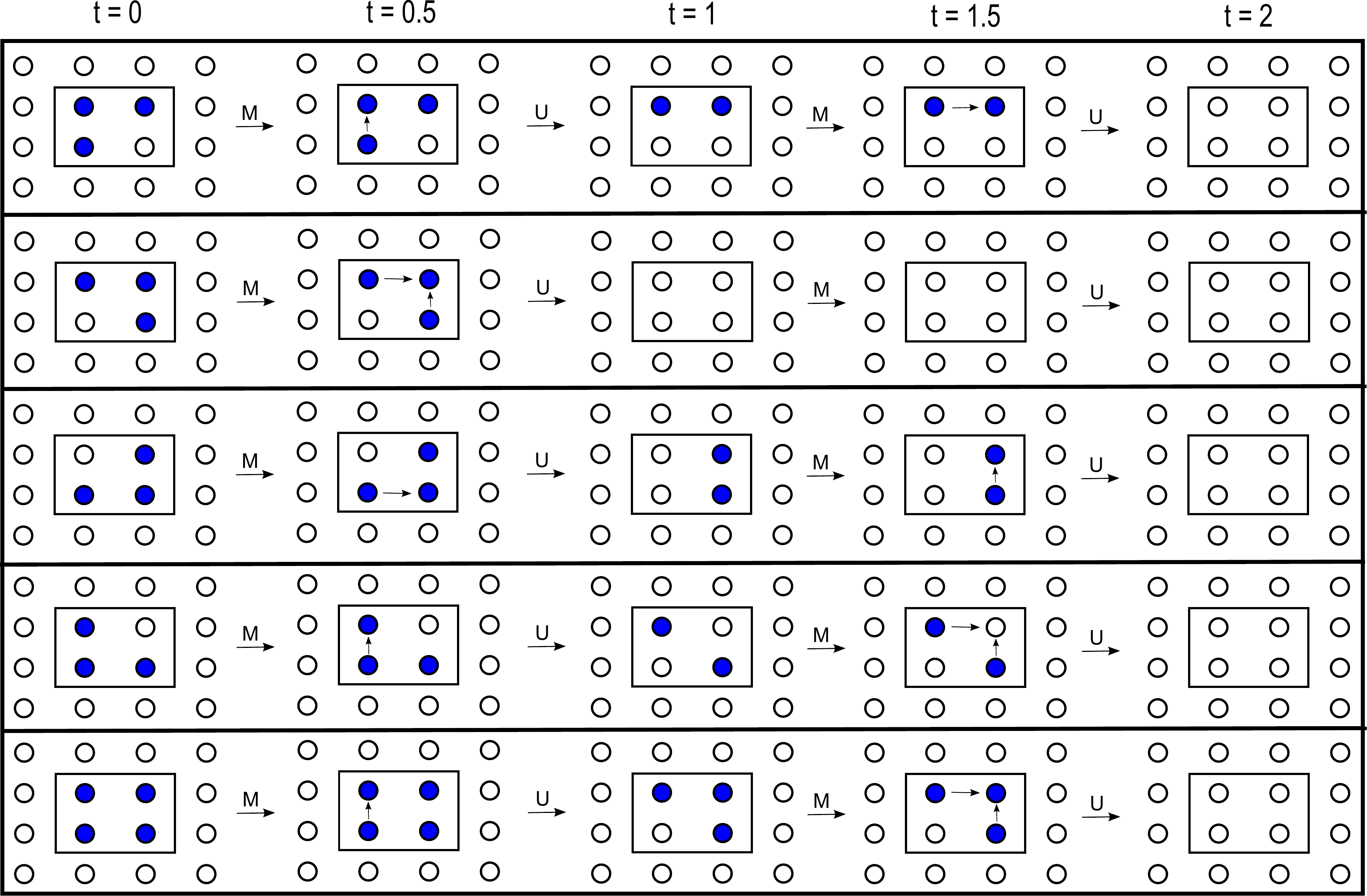}
\caption{Actual level-$0$ errors containing 3 or 4 non-trivial topological charges get corrected within 2 time steps.}
\label{fig_anyon_error1}
\end{figure}

\begin{figure}[!h]
\centering
\includegraphics[width=0.9 \textwidth]{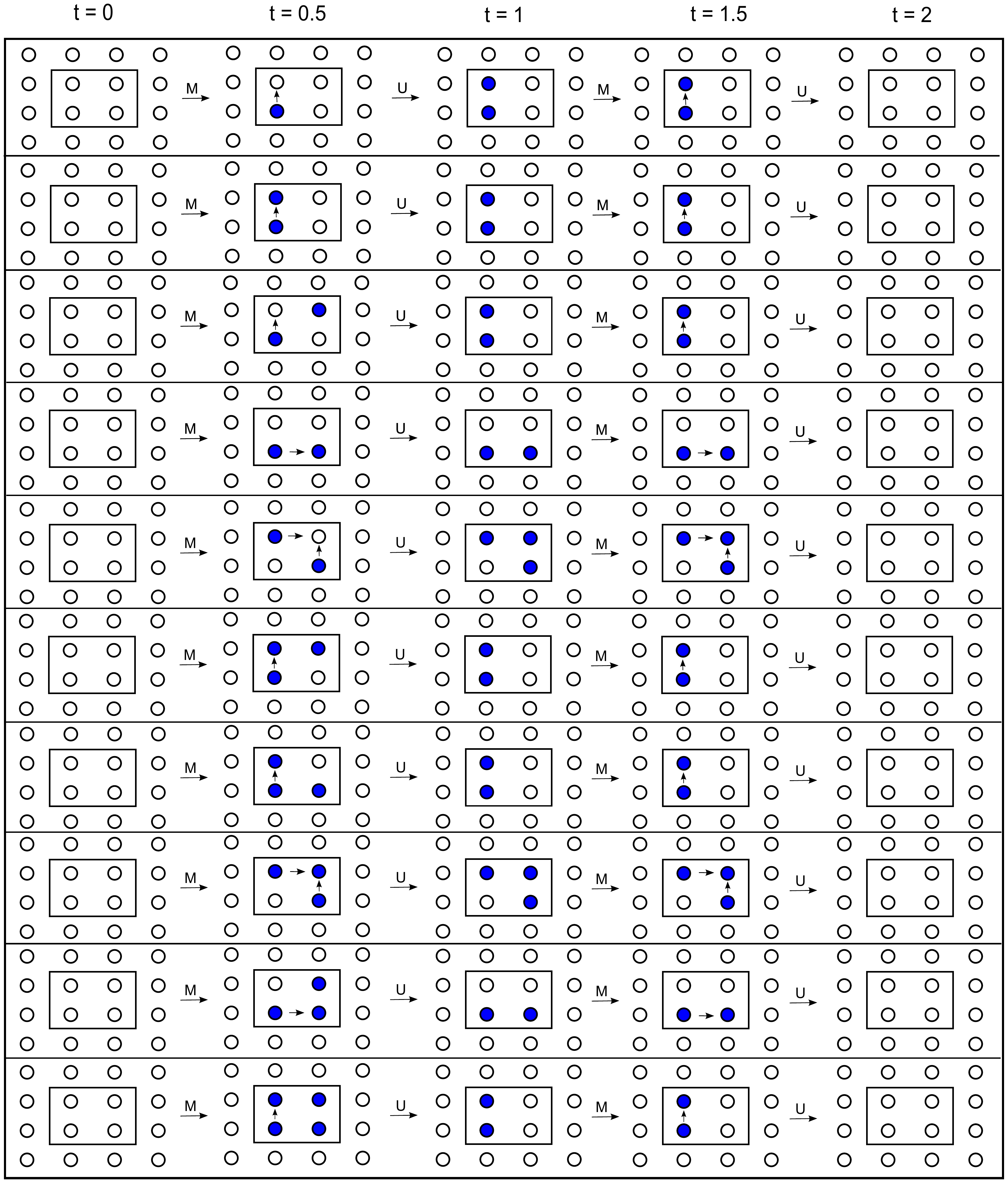}
\caption{Actual level-$0$ measurement errors get corrected within 2 time steps. Again, the precise charges measured do not matter here.}
\label{fig_meas_error}
\end{figure}

\clearpage

The case where a level-$0$ error lies within 2 colonies is handled slightly differently. All possible such errors lying on the West border of a colony are shown below, and they all get corrected within $2$ time steps. By rotational symmetry, the same holds for errors lying on the South border as well.

\begin{figure}[!h]
\centering
\includegraphics[width=.9 \textwidth]{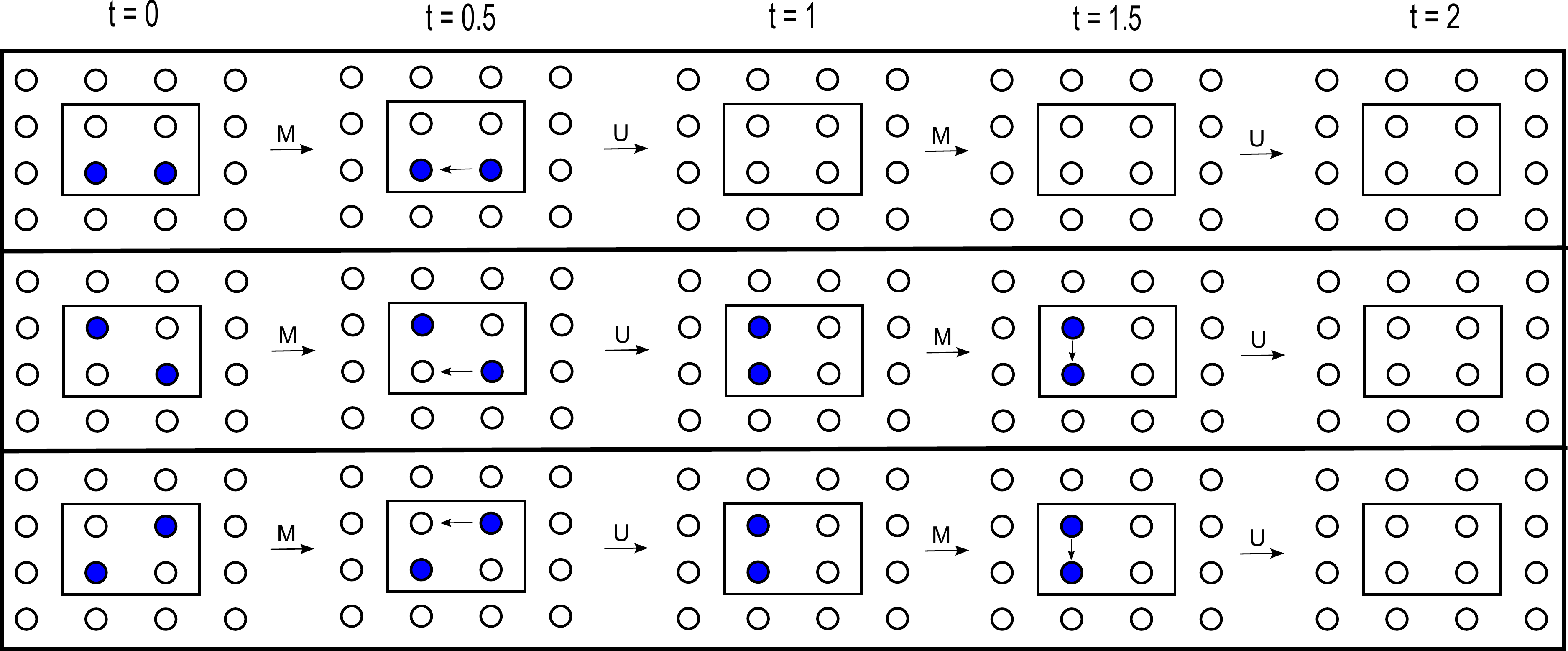}
\caption{Actual level-$0$ errors containing 2 non-trivial topological charges lying on the West border get corrected within 2 time steps.}
\label{fig_charge_border2}
\end{figure}

\begin{figure}[!h]
\centering
\includegraphics[width=.9 \textwidth]{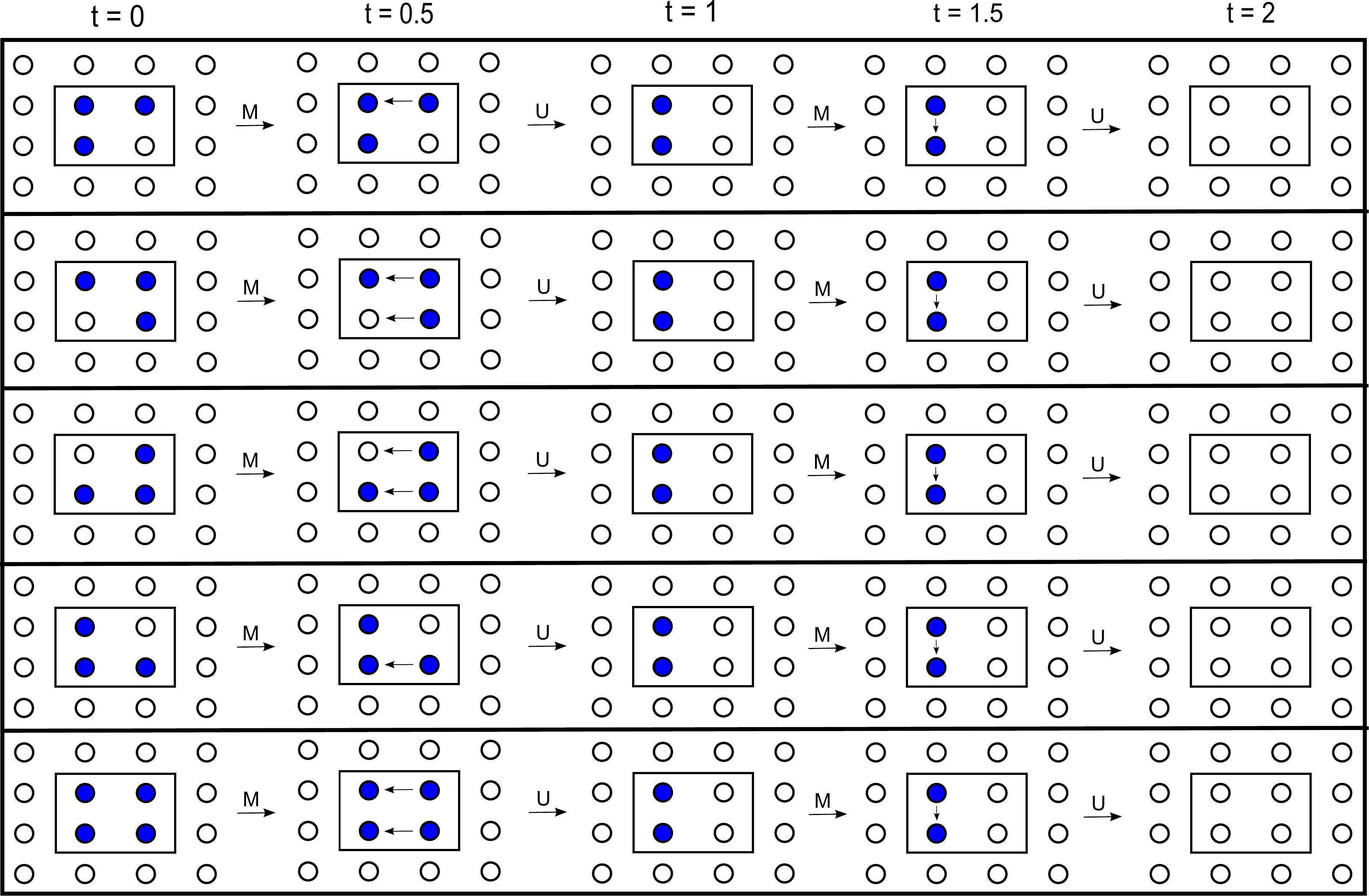}
\caption{Actual level-$0$ errors containing 3 or 4 non-trivial topological charges lying on the West border of a colony get corrected within 2 time steps.}
\label{fig_charge_border3_4}
\end{figure}

\begin{figure}[!h]
\centering
\includegraphics[width=.9 \textwidth]{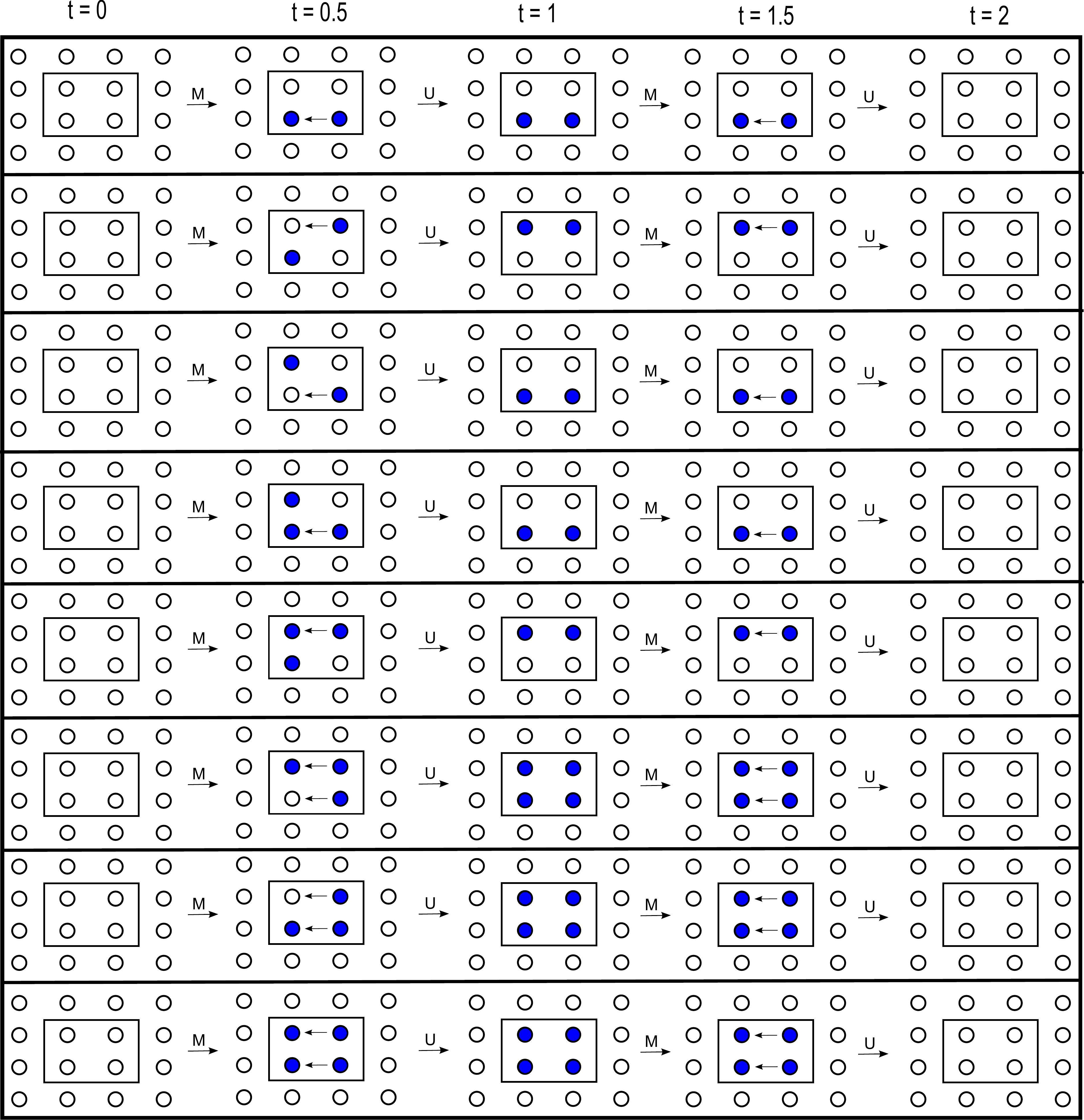}
\caption{Actual level-$0$ measurement errors lying on the West border of a colony get corrected within 2 time steps.}
\label{fig_meas_border}
\end{figure}

\end{document}